\documentclass[conference,compsoc]{IEEEtran}
\usepackage{algorithm,epsfig,cite,stfloats,bm,amssymb,amsmath,color,subfigure,
graphics,extpfeil,amsthm}
\usepackage{algorithmicx}

\newtheorem{pro}{Proposition}
\newtheorem{coro}{Corollary}
\newtheorem{definition}{Definition}

\newtheorem{assumption}{Assumption}
\long\def\symbolfootnote[#1]#2{\begingroup
\def\thefootnote{\fnsymbol{footnote}}
\footnote[#1]{#2}\endgroup}
\psfull

\begin{document}

\title{Secure Ranging with IEEE 802.15.4z HRP UWB}

\author{
\IEEEauthorblockN{
Xiliang Luo, Cem Kalkanli, Hao Zhou, Pengcheng Zhan, and Moche Cohen
}
\IEEEauthorblockA{
Apple, California, USA \\
}
}

\maketitle

\begin{abstract}
Secure ranging refers to the capability of upper-bounding the actual
physical distance between two devices with reliability. This is essential
in a variety of applications, including to unlock physical systems.
In this work, we will look at secure ranging in the context of
ultra-wideband impulse radio (UWB-IR) as specified in IEEE 802.15.4z
(a.k.a. 4z). In particular, an encrypted waveform, i.e. the scrambled
timestamp sequence (STS), is defined in the high rate pulse repetition
frequency (HRP) mode of operation in 4z for secure ranging. 
This work demonstrates the security analysis of 4z HRP when implemented
with an adequate receiver design and shows the STS waveform can enable
secure ranging.
We first review the STS receivers adopted in previous studies and analyze
their security vulnerabilities.
Then we present a reference STS receiver and prove that secure ranging
can be achieved by employing the STS waveform in 4z HRP.
The performance bounds of the reference secure STS receiver are also
characterized. Numerical experiments corroborate the analyses and
demonstrate the security of the reference STS receiver.
\end{abstract}


\section{Introduction}\label{SecIntr}

Ultra-wideband (UWB) impulse radio (IR) uses a series of short pulses to
transmit data. Each pulse in UWB-IR occupies a very small portion of
the time in the order of nanoseconds. UWB-IR is ideal for short-range
communication between devices, precise ranging and location tracking,
and nearby environment sensing \cite{Molish2006Book}.

There have been many studies on how to leverage UWB-IR for ranging
\cite{Sahinoglu2008Book}.
Most of the previous work focused on the theoretical limitation of
ranging accuracy and sought for optimal ranging solutions with UWB-IR.
It was presumed that the receiver could trust the received signal in
the sense that there were no adversarial attackers trying to manipulate
the ranging waveform. This category of ranging is also referred to as
the link-budget optimized ranging \cite{FiRaWP}. In the meantime,
there are many applications which require the integrity of physical
ranging to enable various secure transactions. 
To this end, it is required that the
ranging result, i.e., the measured and reported distance, always
provides an upper bound on the actual physical distance between two
ranging devices, 
a.k.a. distance-bounding \cite{1993db_Principle, 2018Survey_DB}, even when
adversarial attackers are manipulating the ranging waveforms exchanged
between the ranging devices. This type of ranging is referred to as
secure ranging \cite{FiRaWP}, which is the focus of this paper. One example
use case of secure ranging is illustrated in Fig. \ref{fig:secRanging}.

\begin{figure}[t]
	\centering
	\epsfig{file=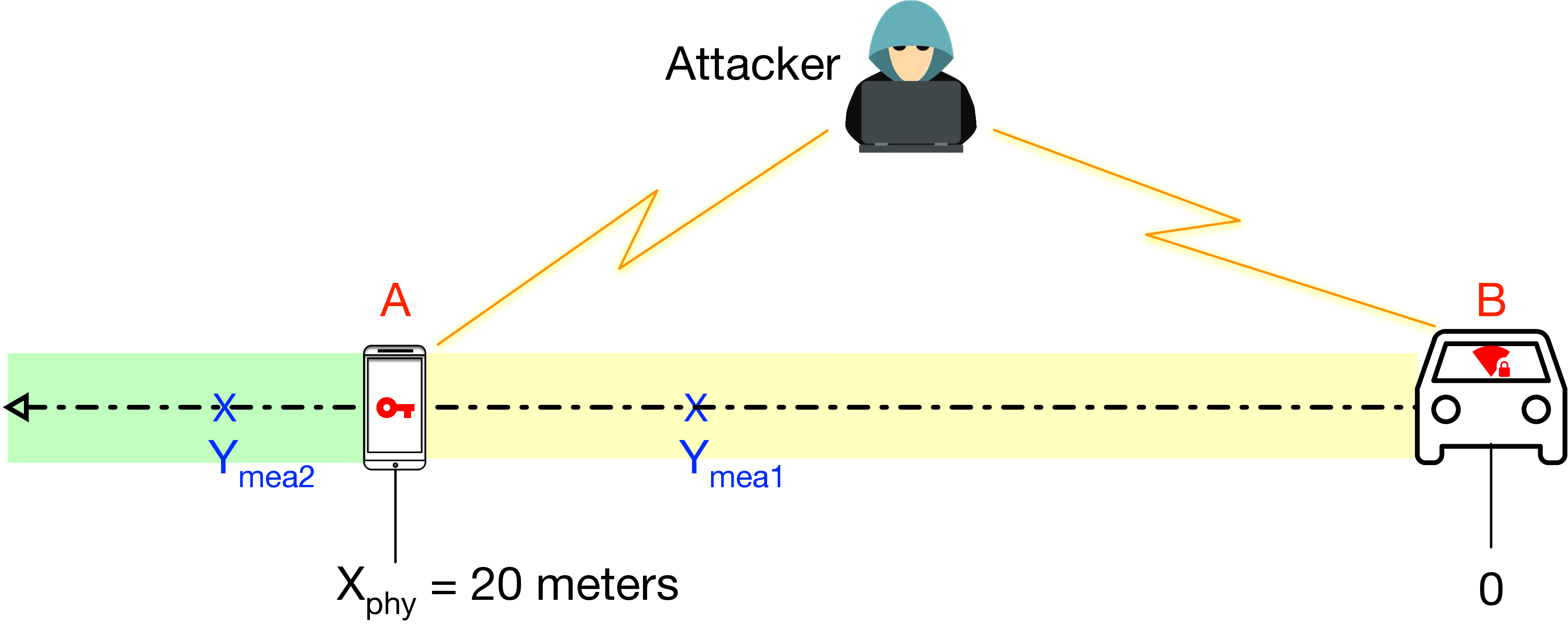, width=0.95\linewidth}
	\caption{A secure ranging use case. Whenever the measured distance from
		the phone-key to the car lies within a prescribed interval, e.g.
		less than $5$ meters, the car will open the door according to the
		protocol. For this protocol to be safe, the phone-key and the car
		need to ensure that the measured and reported distance is always
		lower bounded by the actual physical distance even in the
		presence of an attacker. In other words, the measured
		distance is always an upper bound on the true physical distance
		in the case of secure ranging.
		Distance measurement $Y_{\sf mea2}$
		is a valid one since $Y_{\sf mea2}>X_{\sf phy}$ and thus causes
		no security concerns. On the other hand, an attacker will
		try to reduce the measured distance to $Y_{\sf mea1}$ such
		that $Y_{\sf mea1}<X_{\sf phy}$.}
	\label{fig:secRanging}
\end{figure}

\subsection{Prior Work on Secure Ranging}

The distance-bounding principle was initially proposed by S. Brands and D. Chaum
in \cite{1993db_Principle}. Various distance-bounding protocols were also discussed
in \cite{1993db_Principle} while the physical implementations were left as open
problems. UWB-IR has been recognized as an ideal candidate to realize secure
ranging with distance-bounding since it allows for precise time-of-flight
(ToF) measurements thanks to the short impulses in the order of nanoseconds
\cite{2005RFID_DB}. However, it has been challenging to have a
practical UWB-IR system design with provable secure ranging capability. In
\cite{2010Flury_AttackIR}, the authors revealed the vulnerability of the
ranging solutions based on the standard IEEE 802.15.4a to distance-decreasing
attacks. In particular, the early-decision late-commit (EDLC) attack was
employed in \cite{2010Flury_AttackIR}. Another type of attack called 
``Cicada'' attack was proposed in \cite{2010CicadaAttk} to reduce the distance
measurement in the receiver designed for 802.15.4a and non-line-of-sight (NLoS)
conditions. More studies of physical layer distance-decreasing attacks to 802.15.4a
ranging were presented in \cite{2011TWC_DistBoundWith4a}.
Even though 802.15.4a was not ready for secure ranging, the ToF-based distance
measurement using UWB-IR was concluded to be the way forward in \cite{Capkun2017}.

Amendment 802.15.4z \cite{802.15.4z} specified encrypted ranging waveforms to
increase the integrity and accuracy of the ranging measurements with respect
to 802.15.4a. Ranging security with the low rate pulse repetition frequency
(LRP) mode of operation in 4z was studied in \cite{Capkun2019}. The pulse
repetition frequency (PRF) in 4z LRP mode is less than $4$ MHz \cite{802.15.4z}.
It was shown that the pulse-reordering modulation scheme was able to secure
distance measurement while being robust against all physical-layer distance
shortening attacks \cite{Capkun2019}. Distance-commitment on data is another
scheme to achieve distance-bounding with LRP mode of operation by only capturing
the energy during the short active radio-frequency (RF) period within each symbol
of the data corresponding to the first arriving channel path \cite{802.15.4z}.
In contrast, ranging security is still an open problem in the case
of high rate pulse repetition frequency (HRP) mode of operation in 4z with
PRF greater than $62$ MHz. Due to the higher PRF, the inter-pulse spacing
is less than $16$ nanoseconds in HRP mode, which is less than the delay
spread of typical UWB channels \cite{2009MolischChannel} and incurs severe
inter-pulse interference. On the other hand, the inter-pulse spacing in
LRP mode is more than $250$ nanoseconds, which avoids inter-pulse
interference and thus allows simple schemes like the distance-commitment
on data to achieve ranging security. 
In \cite{Capkun2021_HRP}, the authors provided the first open analysis about
the ranging security with the 4z HRP mode. It was claimed that the HRP mode
of operation was hard to be configured both secure and of satisfactory detection
performance in \cite{Capkun2021_HRP}. In \cite{2022GhostPeak}, the authors
presented over-the-air attacks on several 4z HRP based distance measurement
systems and raised concerns about the usage of 4z HRP in security-critical
applications.

\subsection{Main Contributions}

The HRP mode in IEEE 802.15.4z has been widely deployed in different
solutions \cite{2022UWBStdOverview} due to the hardware implementation
advantages associated with the higher PRF. Particularly, since the
inter-pulse spacing is less than $16$ nanoseconds, a large number
of UWB pulses can be transmitted within each millisecond. The
requirements on the power amplifier (PA) are thus relaxed.
Additionally, the HRP mode is the only one certified by the UWB
certification program in Fine Ranging Consortium (FiRa) \cite{FiRaCert}.
However, there have been no published ranging
solutions based on 4z HRP establishing the security. In this paper, we
demonstrate that secure ranging can be achieved with the STS waveform
in 4z HRP mode. Our primary contributions are as follows.

\begin{enumerate}
\item 
We characterize all the attack waveforms that are theoretically feasible
and provide the rigorous definition of secure ranging in the context of
UWB-IR;
	
\item 
We conduct a review of the previously studied STS receiver for 4z HRP
ranging and elucidate why it cannot provide ranging security as noted
in \cite{Capkun2021_HRP,2022GhostPeak};

\item 
We present a reference STS receiver design for 4z HRP ranging and prove
its security in ensuring the distance-bounding principle. Moreover, the
performance bounds of the reference secure STS receiver are also
characterized.
\end{enumerate}

We are witnessing a proliferation of use cases that involve secure ranging.
The results in this paper will reinforce the foundation.

\subsection{Paper Outline}

Section \ref{SecBackground} provides more details about the HRP mode
in IEEE 802.15.4z.
Section \ref{SecSysModel} describes the system model and defines the
meanings of ``secure ranging'' and ``feasible attack'' in the current
context.
Section \ref{SecCIRThr} reviews the existing STS receiver design which
was widely assumed in previous studies and highlights the effectiveness
of one adaptive attack scheme.
Section \ref{SecSecRx} presents a reference ranging receiver design and
proves the security together with asymptotic optimality in detection.
Performance evaluations are provided in Section \ref{SecSim}.
Finally, concluding remarks are presented in Section \ref{SecConc}.

To facilitate reading, Table \ref{tab:SymbPaper} and Table
\ref{tab:Acronyms} in Appendix \ref{AppendixNotations} list all the
key notations, symbols, and acronyms adopted in this paper.

\section{HRP Mode in IEEE 802.15.4z} \label{SecBackground}

Appendix \ref{AppendixUWBStandard} provides a brief introduction to the
standardization of UWB technology since it was approved by FCC in 2002
\cite{FCC}. The amendment 4z is the focus of the current paper. Note both
LRP and HRP modes are defined in IEEE 802.15.4z. The PRF in 4z LRP
mode is less than $4$ MHz and the inter-pulse spacing is greater than
$250$ nanoseconds. Yet the PRF is greater than $62$ MHz
in HRP mode of operation. Accordingly, the inter-pulse spacing is less
than $16$ nanoseconds. As a result, the UWB pulses in HRP will experience
inter-pulse interference (IPI) after propagating through a multi-path
channel since the inter-pulse spacing is typically less than the delay
spread of the channel \cite{2009MolischChannel}.

\begin{figure}[t]
\centering
\epsfig{file=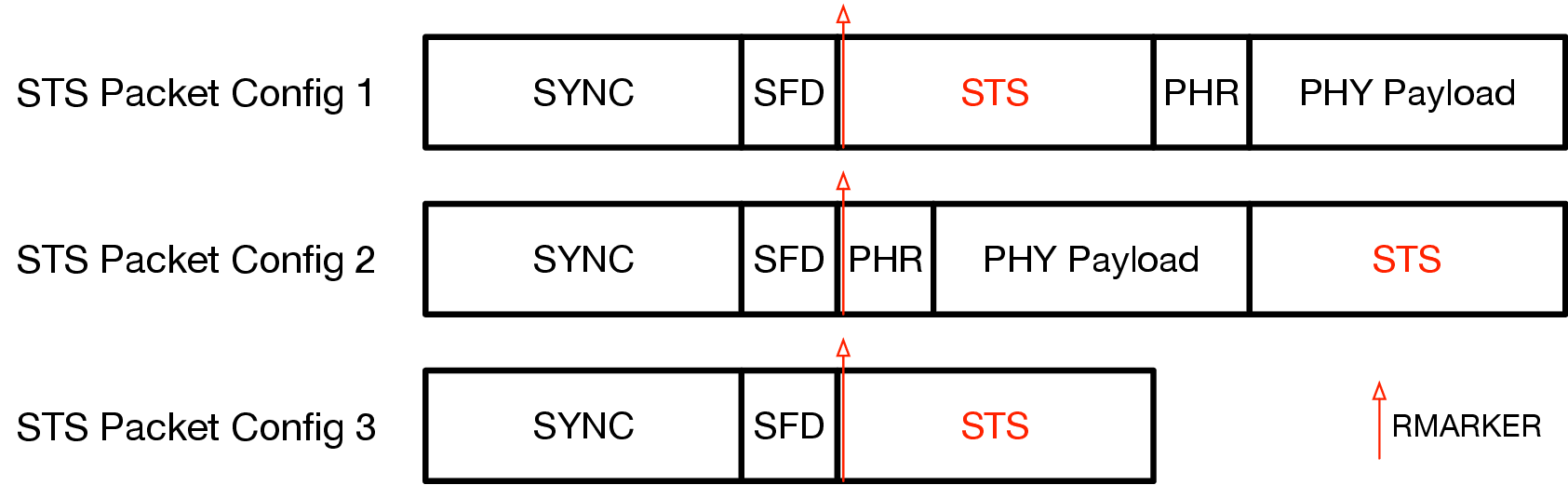, width=0.95\linewidth}
\caption{PHY packet formats for ranging integrity in 4z HRP \cite{802.15.4z}.}
\label{fig:stsPktFormats}
\end{figure}

The PHY packet formats specified in 4z HRP for ranging integrity are
illustrated in Fig. \ref{fig:stsPktFormats}. The designs and functionalities
of individual portions in a PHY packet are as follows.
\begin{itemize}
\item SYNC (Synchronization): used for initial packet acquisition, timing
and frequency synchronization, and channel estimation. In 4z HRP, the SYNC
portion includes a repetition of one ternary preamble code that exhibits an
ideal periodic auto-correlation, i.e. the delta function;
 
\item SFD (Start-of-Frame Delimiter): used to signal the start of a frame
and the end of the SYNC. The SFD in 4z HRP modulates the same ternary code
as in SYNC;

\item PHR (PHY Header): used for signaling of the PHY payload length;

\item PHY Payload: used to convey the data bits from the MAC layer; 

\item STS (Scrambled Timestamp Sequence): used to enable ranging integrity.
The STS portion consists of a sequence of uniformly spaced UWB pulses with
pseudo-random polarities that are generated by applying the AES-$128$
engine in counter mode \cite{AESstd, DRBGstd};

\item RMARKER (Ranging Marker): defined as the time when the peak of the
(hypothetical) pulse associated with the first chip following the SFD is
at the local antenna. All reported times during ranging are measured relative
to the RMARKER.
\end{itemize}

Throughout the remainder of this paper, our focus will be on the
``STS Packet Config 3'' in Fig. \ref{fig:stsPktFormats} and data
communication is outside the scope of the current paper.

The inter-pulse spacing between the STS pulses in 4z HRP is either
$8$ chips for base PRF (BPRF) operation or $4$ chips for higher PRF
(HPRF) operation with the length of each chip being
$T_{c} = 1/499.2$ $\mu$s ($\approx 2$ ns). This inter-spacing
measured in the unit of chips is also known as the spreading factor
(SF).

In 4z HRP mode, the STS portion could be split into multiple segments.
The length of each STS segment in time can be expressed as
$T_{\sf sts} = K_{\sf sts}\times 512T_{c}$, where
$K_{\sf sts}\in\{16, 32, 64, 128, 256\}$. Throughout the rest of the
paper, we will assume a single STS segment unless stated otherwise.

The higher layer configures the $128$-bit STS key. This key is used by
the AES-$128$ engine to generate the cryptographic STS bit sequence
consisting of $0$'s and $1$'s. In 4z HRP, bit $0$ corresponds to a
positive polarity of $+1$, while bit $1$ corresponds to a negative
polarity of $-1$. Appendix \ref{AppendixSTSDRBG} provides more details
about the STS generation structure.

In the SYNC portion, a known preamble code with ideal auto-correlation
characteristics is utilized. This design facilitates the initial packet
acquisition process and ensures that the system can quickly and accurately
establish synchronization. However, it is known that the SYNC is vulnerable
to distance-decreasing attacks. For example, an adversary, e.g., the attacker
depicted in Fig. \ref{fig:secRanging}, could simply transmit the known
preamble code at an advanced timing. This will be able to confuse the
receiver about the correct arrival time of the ranging packet
\cite{2011TWC_DistBoundWith4a}.

In contrast, the attacker cannot predict the STS bit sequence thanks to
its cryptographic nature. This is the key property that underpins secure
ranging even in the presence of attackers. In Section \ref{SecSecRx}, we
will delve into the specifics regarding how the cryptographic STS can be
employed for secure ranging.

\begin{figure}[t]
	\centering
	\epsfig{file=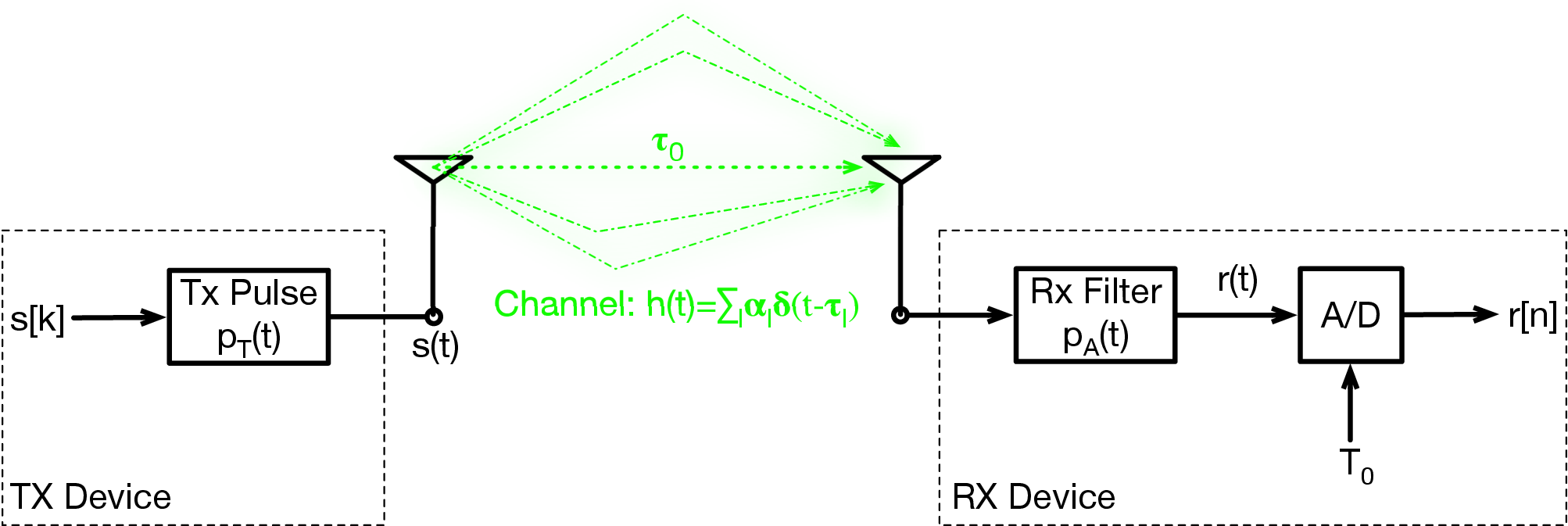, width=0.95\linewidth}
	\caption{System model for STS-based ranging.}
	\label{fig:sysModel}
	\vspace{-3mm}
\end{figure}

\section{System Model} \label{SecSysModel}

In this section, we establish the signal model for the STS-based ranging
in 4z HRP. During one exchange of ranging waveform between ranging devices,
one transmit-device (TXD) sends a ranging packet to one receive-device (RXD).
For example, the phone-key acts as the TXD and the car acts as the RXD in
Fig. \ref{fig:secRanging}. The ranging packet is formed according to 
``STS Packet Config 3'' in Fig. \ref{fig:stsPktFormats}. The RXD receives
the ranging packet through a multi-path wireless channel and determines the
arrival timing of the first path. More mathematical details regarding this
process are given next and Fig. \ref{fig:sysModel} illustrates the overall
system.

\subsection{STS Waveform from TXD}

Let $p_{T}(t)$ denote the baseband transmit pulse satisfying\footnote{
IEEE recommends that the transmitted pulse exhibits minimum precursor
energy \cite{802.15.4z}. In practice, as long as the transmit pulse
exhibits a finite support, we can define the earliest timing of
non-zero energy as $t=0$.} $p_T(t)=0$, $\forall t<0$.
The transmitted STS waveform from the TXD can be expressed as
\begin{equation}\label{txWf}
s(t) = \textstyle\sum_{k=0}^{Q-1} s[k]\cdot p_T(t-k\cdot L\cdot T_{c}),
\end{equation}
where $s[k]\in\{-1, +1\}, k=0,...,Q-1,$ represents the length-$Q$
bipolar STS sequence, $L\in\{4, 8\}$ denotes the spreading factor,
and $T_{c}$ denotes the chip period. Note that
$s(t)$ is the low-pass (a.k.a. baseband) equivalent of the RF signal
emitted from the transmit antenna of the TXD\footnote{Readers can
refer to \cite{DigitalCom} for more information about the equivalence
between the low-pass representation and the corresponding band-pass
RF signal.}.

\subsection{Channel Model}

As commonly used in UWB literature \cite{2009MolischChannel}, we also
adopt the following multi-path channel model in the baseband:
\begin{equation}\label{chModel}
h(t) = \textstyle\sum_{l\ge 0} \alpha_l \delta(t-\tau_l),
\end{equation}
where $\alpha_l$ denotes the complex gain of the $l$-th path with
a delay of $\tau_l>0$ and $\delta(\cdot)$ is the Dirac delta function.
Note $\tau_0$ is the arrival timing of the first path from the TXD
to the RXD.

\subsection{Received Signal at RXD}

After the STS waveform in (\ref{txWf}) propagates through the
channel in (\ref{chModel}), the received baseband signal at
the RXD is
\begin{equation}\label{rxSig1}
r(t) = s(t)\ast h(t)\ast p_A(t) + w(t),
\end{equation}
where the operator ``$\ast$'' stands for the linear convolution
operation, filter $p_A(t)$ captures the overall analog receive filter
impulse response, and $w(t)$ represents the additive noise and
interference. Note that $p_A(t)$ lumps together the effects of
the RF front end and the analog front end at the RXD. Meanwhile,
causality tells us that $p_A(t)=0$, $\forall t<0$.
Let $g(t)$ denote the
aggregate received pulse that combines both the multi-path
channel in (\ref{chModel}) and the overall pulse shape. We have
\begin{eqnarray}
g(t):= p_T(t)\ast h(t) \ast p_A(t) 
     = \textstyle\sum_{l\ge 0} \alpha_l p_R(t-\tau_l),
\label{aggPulse}
\end{eqnarray}
where $p_R(t) := p_T(t)\ast p_A(t)$ denotes the aggregate
pulse without the channel. It can be seen that $p_R(t)=0$,
$\forall t<0$. Note $g(t)$ is referred to as the
channel impulse response (CIR) at the RXD\footnote{The
multi-path model in (\ref{chModel}) is also termed as channel impulse
response sometimes. In this paper, CIR refers to the aggregate pulse
shape $g(t)$.}.
Accordingly, the received signal $r(t)$ in (\ref{rxSig1}) can be
expressed as
\begin{equation}\label{rxSig2}
r(t) = \textstyle\sum_{k=0}^{Q-1}s[k]g(t-k\cdot L\cdot T_c) + w(t).
\end{equation}

The received signal in (\ref{rxSig2}) can be further sampled with
a time period of $T_0 = {T_c}/\Omega$ with $\Omega$ being an
integer greater than or equal to $2$ \cite{DSPbook}.
The resulting digital signal at the RXD is
\begin{eqnarray}
r[n]:=r(nT_0) &=& 
\textstyle\sum_{k=0}^{Q-1}s[k] g(nT_0-k\cdot \Omega LT_0) + w[n], \nonumber\\
&=& \textstyle\sum_{k=0}^{Q-1}s[k] g[n-kM] + w[n], \label{rxSig3}
\end{eqnarray}
where $g[n]:=g(nT_0)$ is the discrete version of the aggregate
received pulse, $M:=\Omega L$, and $w[n]:=w(nT_0)$.

\subsection{Secure Ranging Receiver}

As the first step in secure ranging, the RXD needs to estimate
the first path timing $\tau_0$ in (\ref{chModel}) with the
received signal in (\ref{rxSig3}). Let $\hat{\tau}_0$ denote
the estimate. Following the principle of distance-bounding,
the RXD further needs to reject this timing estimate whenever
it is earlier than the true timing. Before providing a rigorous
definition of secure ranging, we need to clarify what kind of
attacks to the STS-based ranging are feasible. The following
definitions elaborate the meaning of feasible attacks in the
context of the ranging system described above.

\begin{figure}[t]
	\centering
	\epsfig{file=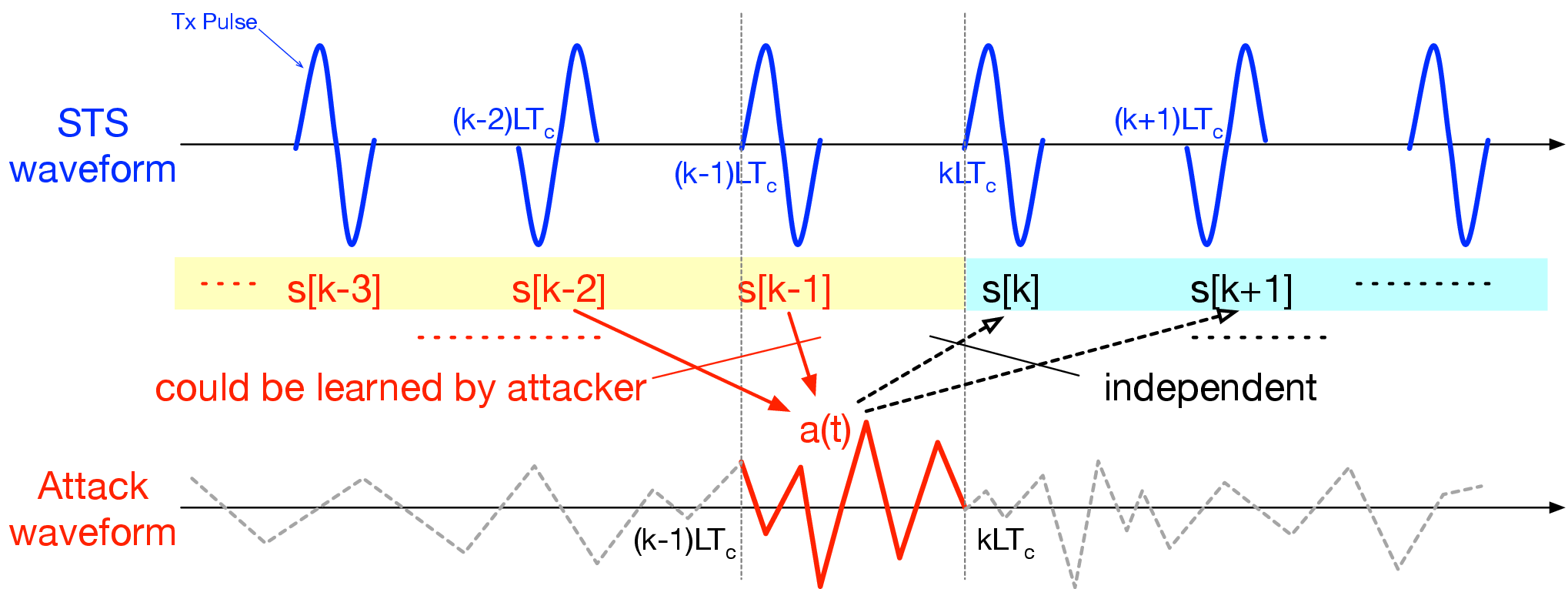, width=0.95\linewidth}
	\caption{Feasible attacks to STS ranging waveform.}
	\label{fig:FeasibleAttack}
\end{figure}

\begin{definition} \label{FeasibleAttk}
	(Feasible Attack)
	An attack to the STS ranging waveform is feasible only if this attack
	does not utilize any non-causal information from the cryptographic STS
	sequence at any time. In other words, $\forall k$, at any time within
	the interval $((k-1)L{T_c}, kL T_c)$, a feasible attack to
	the STS waveform in (\ref{txWf}) can only learn information about
	$\{s[n]|n\le k-1\}$ and is unable to make any inferences about
	$\{s[n]|n\ge k\}$.
\end{definition}

An equivalent way to define feasible attacks is as follows.
\begin{definition} \label{FeasibleAttkAlt}
	Let $a(t)$ represent the time-domain attack waveform generated by an
	attacker. The attack to the STS waveform in (\ref{txWf})
	is feasible only if $a(t)$ is independent of $\{s[n]|n\ge k\}$,
	$\forall t < kL{T_c}$ and $\forall k$.
\end{definition}

Fig. \ref{fig:FeasibleAttack} illustrates feasible attacks according
to Definition \ref{FeasibleAttkAlt}.
Next, we provide a precise definition of secure ranging in the
current context.

\begin{definition} \label{SRdef}
(Secure Ranging)
The ranging receiver at the RXD is secure only if it ensures that
a given estimate of the first path timing $\eta$ is accepted with
a probability no more than a prescribed value whenever it is
earlier than $\tau_0-\Delta$, i.e.,
\begin{equation}\label{SRProbDef}
{\sf Pr}({\sf Accept}~ \eta|\eta <\tau_0 - \Delta) \le \rho,
\end{equation}
where 
$\tau_0$ denotes the true timing of the first path, 
$\Delta>0$ is a constant representing the amount of allowed
implementation headroom, and $\rho$ is the prescribed upper
bound on the false acceptance rate. 

The probability in (\ref{SRProbDef}) is with respect to all the
random realizations of the STS sequences. Furthermore, this
security guarantee also applies to the cases where arbitrary
feasible attacks try to advance the timing estimate by manipulating
the STS ranging waveform from the TXD.
\end{definition}

The above definition builds on the original distance-bounding principle
by S. Brands and D. Chaum in \cite{1993db_Principle} while taking into
account the physical implementations with UWB-IR. In particular, this
definition allows a constant implementation headroom. Taking the use
case shown in Fig. \ref{fig:secRanging} as an example, one secure ranging
solution shall make sure the measured distance always (or with a
probability higher than $1-\rho$) upper bounds the physical distance
as $X_{\sf phy} \le Y_{\sf mea} + \Delta\cdot c$ with $c$ denoting the
speed of light.

To facilitate a rigorous proof of security, we make the following
assumption on the STS sequence $s[k]$ in (\ref{txWf}).
\begin{assumption}\label{Assumption1}
	The STS sequence $\{s[k]\}_{k=0}^{Q-1}$ in (\ref{txWf}) is a sequence
	of independent and identically distributed (IID) binary random variables
	taking value $+1$ or $-1$ with probability $0.5$, i.e.,
	${\sf Pr}(s[k]=+1)={\sf Pr}(s[k]=-1)=0.5$, $\forall k\in[0, Q-1]$.
\end{assumption}

Assumption \ref{Assumption1} indeed states the security requirement
about the STS sequence as well.
Note that the STS key is only shared between two trusted ranging
devices through an encrypted logical channel. The attacker
has no knowledge of the STS key and can not predict future STS
bits by observing the past thanks to the AES construction\footnote{
The IID property of the STS sequence in the above assumption is
based on the premise that only exhaustive key search can break
the AES construction \cite{KatzBook}.
Hoang and Shen recently proved the security and robustness of
the random bit generator that runs AES-$128$ in counter mode
in \cite{2020AES_CTRDRBG}.
} \cite{Capkun2021_HRP}.

We will review different STS receiver designs together with some feasible
attacks in Section \ref{SecCIRThr} and Section \ref{SecSecRx}. 
Furthermore, we will prove the security of the reference STS receiver
in Section \ref{SecSecRx} according to Definition \ref{SRdef}.

\subsection{Clock Frequency Offsets}\label{sSec:ClkFreqOffset}

\begin{figure}[t]
	\centering
	\epsfig{file=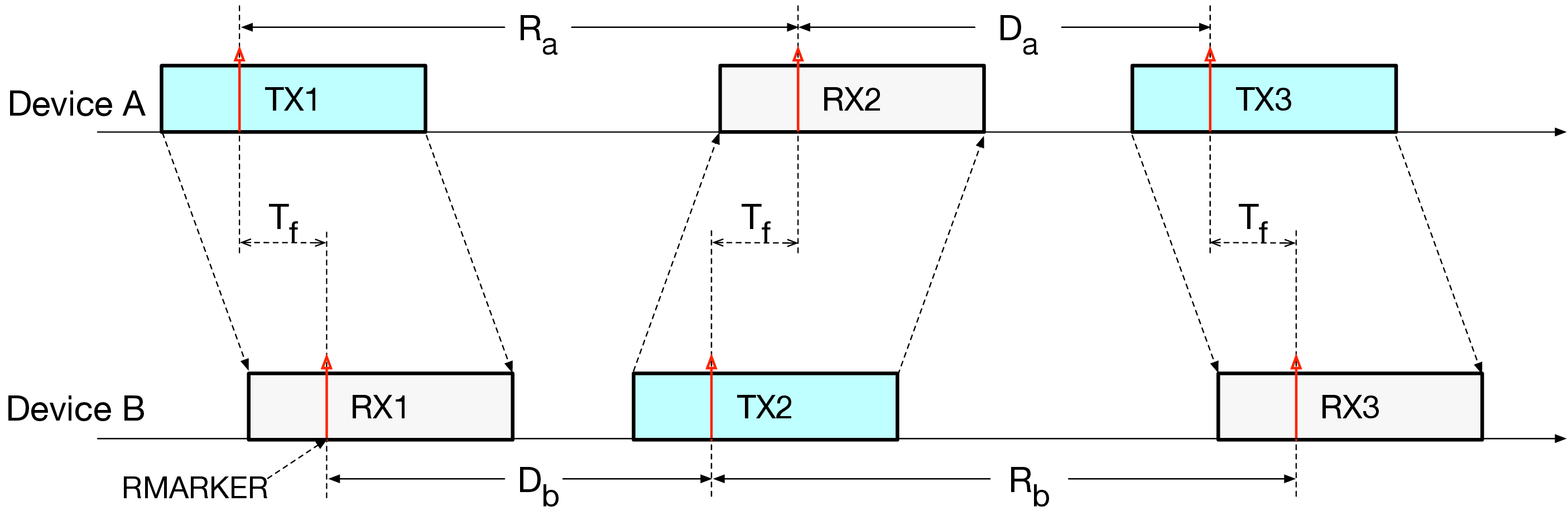, width=0.95\linewidth}
	\caption{Exchanges of ranging packets in double-sided two-way ranging.}
	\label{fig:DSTWR}
\end{figure}

When two devices perform double-sided two-way ranging (DS-TWR), the packet
exchanges and the relevant time measurements are shown in Fig. \ref{fig:DSTWR}.
In particular, $R_a$ and $R_b$ denote the ground truth round-trip times while
$D_a$ and $D_b$ denote the ground truth reply times. Note that all the time
measurements are with respect to the corresponding RMARKERs and $T_f$ represents
the time-of-flight (ToF) from Device A to Device B. Also note that the arrival
timing of the RMARKER can be inferred from the estimate of the first path
timing with the received signal in (\ref{rxSig3}). Due to the clock frequency
offsets at both devices, the measured times can be expressed as
\begin{equation}
	\hat{R}_a = k_a R_a, ~\hat{D}_a = k_a D_a, 
	~ \hat{R}_b = k_b R_b, ~\hat{D}_b = k_b D_b, \nonumber
\end{equation}
where $k_a$ ($k_b$) denotes the ratio of the clock frequency at Device A
(Device B) to the frequency of an ideal clock. As recommended in
\cite{802.15.4z}, the ToF can be estimated as 
\begin{equation}\label{ProdFormula}
	\hat{T}_f = (\hat{R}_a\cdot \hat{R}_b - \hat{D}_a\cdot \hat{D}_b)\big/
	(\hat{R}_a + \hat{R}_b + \hat{D}_a + \hat{D}_b).
\end{equation}
Furthermore, it can be shown that the relative error in $\hat{T}_f$
is only determined by the relative clock offsets \cite{802.15.4z}, i.e.,
\begin{equation}\label{ProdFormulaError}
	(\hat{T}_f-T_f)/{T_f} \approx \big((k_a-1) + (k_b-1)\big)/2.
\end{equation}
Since both $(k_a-1)$ and $(k_b-1)$ are in the order of $\pm 20$ ppm
according to the IEEE requirement \cite{802.15.4,802.15.4z}, we see
the ranging error due to clock offsets is in the order of pico-second.
We can safely neglect this error in practical applications. Accordingly,
as the signal model in (\ref{rxSig3}), we will assume both the TXD
and the RXD have ideal clocks in the sequel.

It is worth mentioning that there is no need to compensate
the local clocks at Device A and at Device B to derive the ToF
estimate in (\ref{ProdFormula}). Similar to the observations
in \cite{2023Usenixsecurity}, we strongly recommend
both Device A and Device B should not try to compensate their
local clocks according to the received waveform to avoid potential
attacks. Further note that the reference STS receiver presented in
Section \ref{SecSecRx} does not compensate the local clocks.

\section{CIR Estimation, Thresholding, and Attack} \label{SecCIRThr}

In this section, we first derive the optimal estimate of the CIR based
on the received signal model in (\ref{rxSig3}). Then we review the STS
receiver that determines the timing of the first path by thresholding
the taps in the estimated CIR as in \cite{Capkun2021_HRP,2022GhostPeak}.
We further present one feasible attack scheme that can create a fake peak
in the CIR at a specified timing. This is distinct from the ``Cicada attack''
in \cite{2010CicadaAttk} and the ``random STS attack'' in
\cite{2022GhostPeak} (a.k.a. ``ghost peak attack'') which can only
create fake peaks at random locations in the CIR.

\subsection{CIR Estimation}

By resorting to the least-squares principle, the RXD can recover the CIR
$g[n]$ in (\ref{rxSig3}) as in the following proposition. More details
can be found in Appendix \ref{AppendixCIREst}.

\begin{pro}\label{PropCIREst}
Assume the support of $g[n]$ is limited to $[0, JM-1]$, i.e.,
$g[n]=0$, $\forall n\in(-\infty, 0)\cup [JM, +\infty)$, where
$JM$ characterizes the maximum delay spread of $g[n]$.
Let ${\bm g}$ denote the vector version of the CIR, i.e.,
\begin{equation}\label{vecCIR}
{\bm g} := \left[g[0],...,g[JM-1]\right]^T.
\end{equation}
Given the signal model in (\ref{rxSig3}), the RXD can derive the following
CIR estimate with the knowledge of $\{s[k]\}_{k=0}^{Q-1}$:
\begin{equation}\label{slcCIR}
\hat{\bm g} 
= \left({\bm \Phi}^T{\bm \Phi}\right)^{-1}{\bm \Phi}^T{\bm r}
:= \left({\bm \Phi}^T{\bm \Phi}/Q\right)^{-1}\tilde{\bm g},	
\end{equation}
where 
${\bm r}:=[r[0],...,r[(Q-1+J)M-1]]^T$,
$\tilde{\bm g}:={\bm \Phi}^T{\bm r}/Q$,
and the matrix $\bm \Phi$ is a
Toeplitz matrix of size $[(Q-1+J)M]\times [JM]$. 
Specifically, the first row of $\bm \Phi$ is given by
$ {\bm \Phi}(1,:) = [s[0], {\bm 0}_{JM-1}]$
and the first column of $\bm \Phi$ is given by
${\bm \Phi}(:,1) = [s[0], {\bm 0}_{M-1}, 
s[1], {\bm 0}_{M-1}, ..., s[Q-1], {\bm 0}_{M-1},
{\bm 0}_{(J-1)M}]^T$.
The notation ${\bm 0}_{a}$ stands for a row vector of
all $0$'s of size $1$-by-$a$.
\end{pro}

Note that $\tilde{\bm g}$ in (\ref{slcCIR}) represents the intermediate
CIR estimate resulting from the simple correlation of the STS sequence with
the received signal \cite{2022GhostPeak,Capkun2021_HRP}. To cancel the
interference among different CIR taps due to the non-ideal auto-correlation
of STS sequence, additional matrix multiplication and inverse operations
are performed to derive a cleaner estimate $\hat{\bm g}$ as shown in
(\ref{slcCIR}).

\subsection{First Path Validation via CIR Thresholding}

In the absence of adversarial attacks, the estimated CIR $\hat{\bm g}$
according to Proposition \ref{PropCIREst} will be a superposition of the
actual CIR: ${\bm g}$ and the estimation error: $\bm e$, i.e.,
$\hat{\bm g} = {\bm g} + {\bm e}$. After learning about the standard
deviation $\sigma$ of the samples in $\bm e$, one CIR tap could be validated
as a true one when the amplitude of the tap exceeds a certain acceptance
threshold with respect to the learned standard deviation according
to the desired false acceptance rate. One example is illustrated in
Fig. \ref{fig:CIRThreshold} in Appendix \ref{AppendixCIRThr} where the
threshold is set to $3\sigma$ to realize a false acceptance rate around
$0.1\%$ under the Gaussian distribution.

\begin{figure}[t]
	\centering
	\epsfig{file=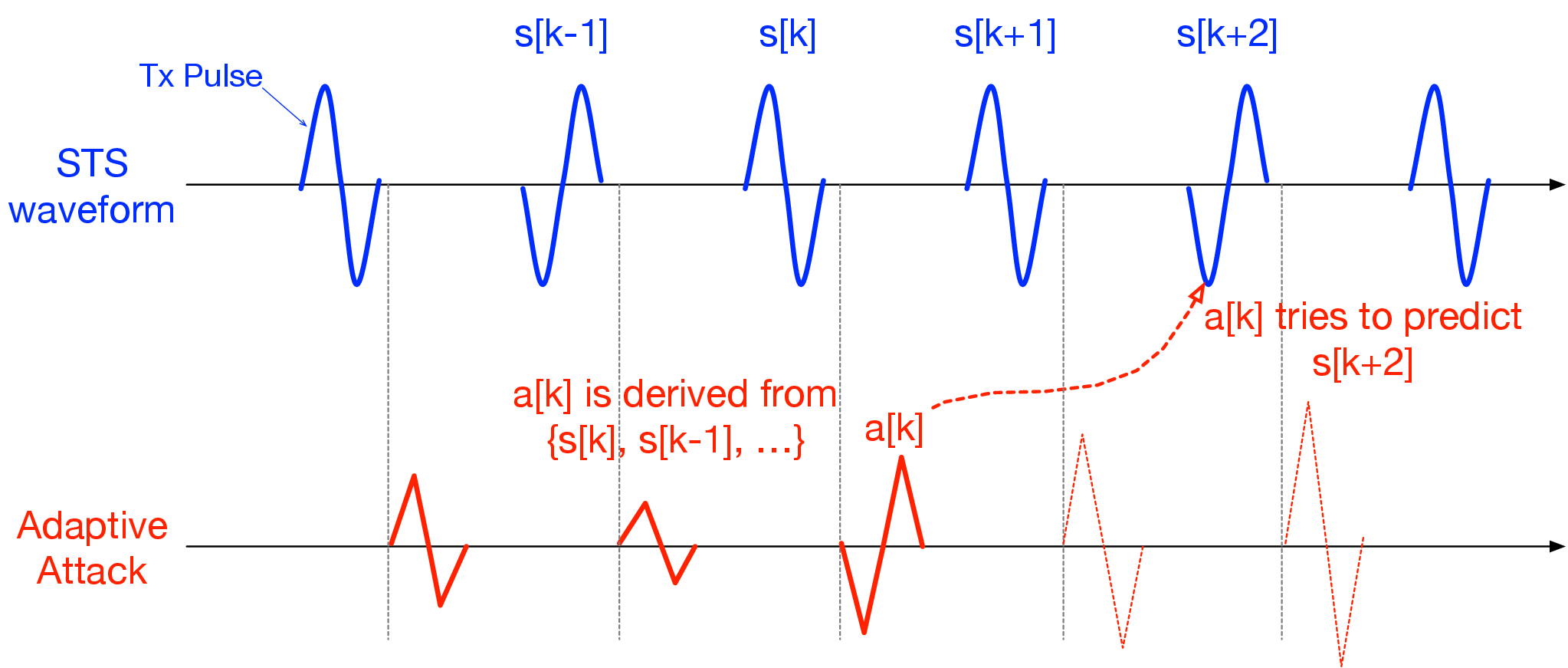, width=0.95\linewidth}
	\caption{An adaptive scheme to attack the STS waveform.}
	\label{fig:AdaptiveAttack}
\end{figure}

On the other hand, the authors in \cite{2022GhostPeak, Capkun2021_HRP}
showed that the estimated CIR could exhibit fake peaks earlier
than the true timing of the first path when an adversary attacks the
STS waveform. In particular, either the ``Cicada attack'' in
\cite{2010CicadaAttk} or the ``random STS attack'' in
\cite{2022GhostPeak} (a.k.a. ``ghost peak attack'') could be utilized.
Fake early peaks are also illustrated in Fig. \ref{fig:CIRThreshold}.
It has been recognized that it is hard
to configure the acceptance threshold to strike a balance between the
false acceptance rate in the case of attacks and the detection
performance in validating the true first path when it is weak.

Neither the ``Cicada attack'' nor the ``random STS attack'' are able
to create fake peaks at given locations in the CIR recovered by the RXD.
Next, we will show one particular feasible attack scheme that is able
to generate peaks at the designated locations in the CIR based on the
targeted distance reduction.

\subsection{Adaptive Attacking Scheme}

In order to create a peak at a designated location in the CIR, one
adversary can follow the procedure in Algorithm \ref{algoAdapAttk}.
Although the history STS symbols do not carry any information about
the future ones according to Assumption \ref{Assumption1},
the adversary could try to predict the future STS symbol $s[k+2]$ by
observing the legitimate STS waveform in a causal manner as illustrated
in Fig. \ref{fig:AdaptiveAttack}.
To this end, the adversary simply updates the empirical estimates of the
STS correlation at lags $\{\Lambda, \Lambda+1,..., \Lambda+H\}$
with the observed STS symbol $s[k]$ as in Step-1 of Algorithm
\ref{algoAdapAttk}.

In Step-2, the adversary exploits these updated empirical estimates
of the STS correlation to synthesize the attack symbol $a[k]$ as in
(\ref{attkSymb}). The idea is
to utilize the empirical STS correlation at lag $\Lambda+\lambda$,
i.e., ${\sf corr}[k; \Lambda+\lambda]$, to infer information about
the future STS symbol $s[k+\Lambda]$ from the history STS symbol
$s[k-\lambda]$ which has been observed. 
We have aggregated $(1+H)$ inference terms from
$(1+H)$ history STS samples, i.e., $\{s[k-H], s[k-H+1], ..., s[k]\}$
in (\ref{attkSymb}). In Fig. \ref{fig:AdaptiveAttack}, the attack symbol
$a[k]$ tries to predict STS symbol $s[k+2]$. The step size of
prediction in Fig. \ref{fig:AdaptiveAttack} is thus $\Lambda=2$.

In Step-3, the adversary transmits the attack waveform by modulating
the transmit pulse $p_T(t)$ with the synthesized attack symbols.
Note that we have assumed the adversary employs the same transmit
pulse as in (\ref{txWf}) to obtain (\ref{attkSig}). This is
for the sake of conciseness in the following derivations. Other
pulse shapes are also allowed to realize this adaptive attack.

\begin{algorithm}[t]
\caption{Adaptive STS Attack on a Designated CIR Tap}
\label{algoAdapAttk}
\begin{algorithmic}[1]
\item {\bf INPUT}: 
\begin{itemize}
\item[-] $\Lambda>0$: the step size of empirical prediction;

\item[-] $H\ge 0$: the number of additional history samples utilized
during the empirical prediction.
\end{itemize}

\item {\bf DO} the following steps after learning about the current STS
symbol $s[k]$
\begin{itemize}
	\item[-] {\bf Step-1}: update the empirical STS correlation value
	${\sf corr}[k; \Lambda+\lambda]$, $\lambda\in[0,H]$, as
	follows.
	\begin{equation}
	{\sf corr}[k; \Lambda+\lambda] = \textstyle\sum_{l=\Lambda+H}^{k}
	s[l]s[l-\Lambda-\lambda];
	\end{equation}

	\item[-] {\bf Step-2}: determine the attack symbol $a[k]$ as follows.
	\begin{equation}\label{attkSymb}
	a[k] = \textstyle\sum_{\lambda=0}^{H} s[k-\lambda]\cdot 
	{\sf corr}[k; \Lambda+\lambda];
	\end{equation}
	
	\item[-] {\bf Step-3}: modulate the pulse with $a[k]$ and transmit the
	following attack waveform:
	\begin{equation} \label{attkSig}
	a(t) = \textstyle\sum_{k=0}^{Q-1} a[k]\cdot p_T(t-k\cdot L \cdot T_c).
	\end{equation}
\end{itemize}
\item {\bf END}
\end{algorithmic}
\end{algorithm}

For the sake of clarity, we assume the attack waveform $a(t)$
in (\ref{attkSig}) propagates through a single path channel\footnote{
Due to laws of physics, the overall delay of the attack waveform
is always lower bounded by the delay of the first path in the
channel from the TXD to the RXD for the attacker to	learn the
STS samples causally.} with gain $\theta$ and a delay of $\tau_0$.
Accordingly, the received signal at the RXD becomes
\begin{equation}\label{rxSigAttk}
r_a(t) = \theta \textstyle\sum_{k=0}^{Q-1}a[k] p_R(t-k\cdot L{T_c}-\tau_0).
\end{equation}
We have neglected the additive
noise term in (\ref{rxSigAttk}). Furthermore, we assume that the
legitimate signal is overwhelmed by the attack one. Accordingly, the
legitimate signal is also neglected here.

After sampling the signal in (\ref{rxSigAttk}) with a time period of
$T_0 = T_c/\Omega$ as in (\ref{rxSig3}), we have
\begin{eqnarray}
r_a[n]:=r_a(nT_0)= \theta \textstyle\sum_{k=0}^{Q-1}a[k] p_R[n-kM-\bar{\tau}_0] ,
\label{rxSigAttk2}
\end{eqnarray}
where $\tau_0$ is assumed to be an integer multiple of $T_0$, i.e.,
$\tau_0:= \bar{\tau}_0 T_0$, and $p_R[n]:=p_R(nT_0)$.

Given the received signal in (\ref{rxSigAttk2}), the RXD can
recover the CIR according to Proposition \ref{PropCIREst} as follows.
\begin{eqnarray}
	\widetilde{\sf CIR} &=&{\bm \Phi}^T{\bm r}_a/Q, \label{estCIRattk1}\\
	\widehat{\sf CIR} &=& \left({\bm \Phi}^T{\bm \Phi}/Q\right)^{-1}
	\widetilde{\sf CIR},
	\label{estCIRattk2}	
\end{eqnarray}
where ${\bm r}_a:=[r_a[0],...,r_a[(Q-1+J)M-1]]^T$.
Now, we can establish the following proposition regarding the
adaptive attack scheme in Algorithm \ref{algoAdapAttk}.

\begin{pro}\label{PropAdapAttk}	
Let $\bar{\tau}_0$ denote the actual first path timing in the
channel between the attacker and the RXD as in (\ref{rxSigAttk2}).
Assume $p_R(t)$ is confined within $[0, MT_0)$.
The attack waveform as specified in Algorithm \ref{algoAdapAttk}
is able to create a peak at
$l_* := (\bar{\tau}_0 + \epsilon- \Lambda M)$,
in the CIR recovered by the RXD as in (\ref{estCIRattk2}),
where $\epsilon$ could take any integer in $[0, M-1]$.
Furthermore, the mean value of the tap at $l_*$ is a multiple of the
actual channel gain as $Q$ gets large, i.e.,
\begin{equation}\label{meanCIRlimit}
\lim_{Q\rightarrow\infty} {\sf E}\left[\widehat{\sf CIR}_{l_*}\right]
=-\frac{H+1}{2} \theta \cdot p_R[\epsilon],
\end{equation}
where the expectation is with respect to the random realizations
of the STS sequence under Assumption \ref{Assumption1}.
\end{pro}

Appendix \ref{AppendixAdapAttk} outlines the proof of the above
proposition. Numerical results corroborating the result in
(\ref{meanCIRlimit}) are presented in Section \ref{SubSec_SimAdapAttk}.

\subsection{More Discussions}

From prior discussions, it is evident that developing a secure
ranging capability for an STS receiver is challenging when
relying on comparing the taps in the CIR with a particular threshold.
Furthermore, an adaptive attack scheme has been presented that
can generate a fake peak at a specific location in the CIR.
This is due to the fact that the receiver wants to obtain a
clean CIR estimate from the received STS ranging waveform.
One strong tap could easily overshadow the weak one due to the
inter-tap interference. The receiver thus needs to cancel the
interference from the strong taps to reveal the weak one.
It turns out the interference cancellation process at the receiver
could be exploited by the adversary to reduce the distance measurement.
It should be noted that the receiver is unaware of whether the
received waveform is being manipulated. 

We share the same view as in \cite{Capkun2021_HRP, 2022GhostPeak}
that the simple CIR thresholding based approaches
are not recommended for applications where secure ranging is required.
However, the encrypted STS ranging waveform in 4z HRP mode of operation
is intact. In the following section, we will present a different STS
receiver design and prove its security according to Definition \ref{SRdef}.

\begin{algorithm}[t]
	\caption{STS Receiver with Proved Security for 4z HRP}
	\label{algoSecRx}
	\begin{algorithmic}[1]
		
		\item {\bf INPUT}: 
		\begin{itemize}
			\item[-] $l_*$: index of a CIR tap to be validated; 
			
			\item[-] $\{\mathring{r}[n]\}$: the SYNC samples earlier than
			the first STS sample that arrives at the timing corresponding
			to the CIR tap $l_*$ in the current STS packet;
			
			\item[-] $\mathring{g}[n]$, $n\in[0, JM-1]$:
			the CIR estimate with the SYNC samples $\{\mathring{r}[n]\}$;
			
			\item[-] $\mathring{\beta}>0$: a scaling parameter
			according to the signal power level of the received
			SYNC samples $\{\mathring{r}[n]\}$;
			
			\item[-] $\gamma>0$: a detection threshold.
		\end{itemize}
		
		\item {\bf DO} the following steps to decide whether the tap
		$l_*$ should be accepted or rejected
		
		\begin{itemize}
			\item[-] {\bf Cancellation}: Cancel the inter-tap interference from
			$y[n]$ in (\ref{yn}) with the the CIR estimate $\mathring{g}[n]$,
			i.e.,
			\begin{equation}\label{intfCancel}
				\tilde{x}[n] = 
				y[n] - \textstyle\sum_{\zeta > 0} \mathring{g}[l_*+\zeta M]s[n-\zeta];
			\end{equation}
			
			\item[-] {\bf Saturation}: Scale the signal $\tilde{x}[n]$ and then
			input to a non-linear saturation function, i.e.,
			\begin{equation}\label{saturation}
				x[n] = {\mathbb Q}\left({\sf Re}\left\{
				\mathring{\beta} \cdot \mathring{g}[l_*]^{\dagger}\cdot
				\tilde{x}[n]
				\right\}
				\right),
			\end{equation}
			where ${\sf Re}\{\cdot\}$ returns the real part of the argument,
			$(\cdot)^{\dagger}$ is the conjugate operation, and
			the function ${\mathbb Q}(\cdot)$ is defined as
			\begin{equation}\label{funcQ}
				{\mathbb Q}(x)=
				\begin{cases}
					-1 & \text{if } x \le -1 \\
					~x & \text{if } x \in(-1, +1) \\
					+1 & \text{if } x \ge +1
				\end{cases};
			\end{equation}
			
			\item[-] {\bf Validation}: 
			Compute the decision metric as
			\begin{equation}\label{decMetric}
				T({\bm x}) = \textstyle\sum_{n=0}^{Q-1} x[n]s[n]\big/{\sqrt{Q}},
			\end{equation}
			and validate CIR tap $l_*$ according to the following criterion:
			\begin{equation}\label{validation}
				\begin{cases}
					\text{Accept $l_*$} & \text{if } T({\bm x}) \ge \gamma \\
					\text{Reject $l_*$} & \text{if } T({\bm x}) < \gamma
				\end{cases},
			\end{equation}
			where $\gamma$ is a threshold.
		\end{itemize}
		
		\item {\bf END}
	\end{algorithmic}
\end{algorithm}

\section{A Reference STS Receiver with Proved Security} \label{SecSecRx}

\begin{figure*}[t]
	\begin{eqnarray}
		\label{yn}
		y[n]:= r[l_*+nM]
		&=& \textstyle\sum_{k=0}^{Q-1} s[k]g[l_*+nM-kM] + w[l_*+nM]
		= \textstyle\sum_{k} g[k]s_e[l_*+nM-k] + w[l_*+nM] \nonumber\\
		&=& g[l_*]s[n] + \textstyle\sum_{\zeta \neq 0} g[l_*+\zeta M]s[n-\zeta] 
		+\dot{w}[n]
	\end{eqnarray}
	\hrulefill
\end{figure*}

As shown in Proposition \ref{PropCIREst}, the RXD can first obtain an
estimate of the CIR after observing the signal $r[n]$ in (\ref{rxSig3}).
Next the RXD identifies a particular tap in the estimated CIR as the
first path candidate and will utilize the corresponding timing to derive
the associated ToF value. A critical task here is to either accept this
candidate tap as a true physical path or reject it. In the case of attack,
there could be fake peaks in the estimated CIR that are earlier than
the true physical first path. A secure STS receiver needs to reject
any fake peaks earlier than the true first path reliably. Specifically,
whenever validating the CIR tap $l_*$, the RXD essentially determines
one of the following two hypotheses:
\begin{itemize}
\item ${\cal H}_0$: tap $l_*$ does not correspond to any true physical
path;

\item ${\cal H}_1$: tap $l_*$ corresponds to a real physical path.
\end{itemize} 
The probability of accepting ${\cal H}_1$ when ${\cal H}_0$ is true,
i.e., ${\Pr}({\cal H}_1; {\cal H}_0)$, dictates the security of the
receiver when the candidate tap $l_*$ is earlier than the true first path. 
Meanwhile, the probability of deciding ${\cal H}_0$ when ${\cal H}_1$
is true, i.e., ${\Pr}({\cal H}_0; {\cal H}_1)$, characterizes
the miss detection rate when a true physical path is indeed present
at $l_*$.

Following the signal model in (\ref{rxSig3}) and its alternative
expression in (\ref{expandedSig}), we can extract the signal samples
that convey information about the tap $l_*$ as in (\ref{yn}),
where the first term in (\ref{yn}) comes from the tap $l_*$,
the second term captures the inter-tap interference,
$\dot{w}[n]:=w[l_*+nM]$, and $s_e[k]$ denotes the expanded
version of $s[k]$, i.e., $s_e[k] := s[k/M]$,
$\forall k=0,\pm M, \pm 2M, ...$ and $s_e[k]=0$ otherwise.

In the following subsections, we will first present a reference
STS receiver and then show how this design guarantees ranging
security.

\subsection{Reference STS Receiver Design}

We have illustrated how to validate a given CIR tap $l_*$ in Algorithm
\ref{algoSecRx}. In particular, the RXD first processes the SYNC
before utilizing the STS waveform for ranging security. 
Note the SYNC is ahead of the STS as depicted in
Fig. \ref{fig:stsPktFormats}. Accordingly, the receiver design
in Algorithm \ref{algoSecRx} makes use of the CIR estimate from
the received SYNC. It is worthwhile to mention that the CIR
estimate with the SYNC could follow the same principle as in
Proposition \ref{PropCIREst}. Specifically, the ideal auto-correlation
property of the ternary preamble codes in the SYNC renders
${\bm \Phi}^T{\bm \Phi} \propto {\bm I}_{JM}$, where ${\bm I}_{JM}$
stands for an identity matrix of size $JM\times JM$. As a result,
the need of matrix inversion in (\ref{slcCIR}) is avoided.

The ``Cancellation'' step in Algorithm \ref{algoSecRx} is
straightforward according to the signal model in (\ref{yn}).
However, unlike conventional designs, only those taps later
than tap $l_*$ are canceled in (\ref{intfCancel}).
Albeit minor, this turns out to be an important contribution to the
security of the design in Algorithm \ref{algoSecRx}. More
analyses will be presented in Section \ref{sSec_Security}.

During the step of ``Saturation'', we first scale the post-cancellation
signal before performing the saturation. This is to optimize the
detection performance when tap $l_*$ corresponds to a real physical
path while maintaining the specified security level. More details
on this will be illustrated in Section \ref{sSec_DetPerf}.

The decision metric $T({\bm x})$ in (\ref{decMetric}) is compared
with a detection threshold $\gamma$ in the step of ``Validation''
to make a decision regarding the tap $l_*$. This value of this
threshold should be configured according to the desired security
level as explained in Section \ref{sSec_Security}.


\subsection{Security Analysis}\label{sSec_Security}

In Section \ref{sSec:ClkFreqOffset}, we have justified that we can
safely assume both the TXD and the RXD have ideal clocks. In the
following security analysis, we will first look into the scenario
where the clock at the attacker ticks at the same rate as the TXD
clock. Section \ref{sSec:ClockOffsetAttk} will examine how an attacker
could take advantage of the clock offsets relative to the TXD. 

To analyze the security of the reference STS receiver in Algorithm
\ref{algoSecRx}, we can just focus on the case when the CIR tap $l_*$
is earlier than the first path in (\ref{yn}). When
$l_*<\tau_0/T_0$, according to (\ref{aggPulse}), we see
$g[l_*+\zeta M]=0$, $\forall \zeta \le 0$, and $y[n]$ in (\ref{yn})
becomes
\begin{eqnarray}
y[n] = \dot{w}[n] + y_{\sf legit}[n],
\label{H0}
\end{eqnarray}
where $ y_{\sf legit} := \sum_{\zeta > 0} g[l_*+\zeta M]s[n-\zeta]$
and $\dot{w}[n]$ models both the adversarial attack and the additive
noise. According to Definition \ref{FeasibleAttkAlt}, for any 
feasible attack, the received waveform from the attacker at time $n$,
i.e., $\dot{w}[n]$ in (\ref{H0}), is independent to the STS sequence
starting from time $n$, i.e., $\{s[k]|k\ge n\}$. Note the information
about $s[n]$ is only available when the actual first path conveying
$s[n]$ arrives at the RXD. Also note that the legitimate signal
$y_{\sf legit}[n]$ is also independent of $\{s[k]|k\ge n\}$.

Following the steps specified in Algorithm \ref{algoSecRx}, we can
verify that the signal after cancellation in (\ref{intfCancel})
and the saturated signal in (\ref{saturation}) are both independent
of the STS sequence: $\{s[k]|k\ge n\}$. Specifically, the estimate
of CIR tap $\mathring{g}[n]$ is based on the received SYNC that
arrives earlier than the STS. No matter whether the adversary also
attacks the SYNC or not, we can see that the tap estimate
$\mathring{g}[n]$ from the SYNC is always independent of
the STS sequence. Similarly, the scaling parameter $\mathring{\beta}$
is also independent of the STS sequence given that it is determined
from the received SYNC as well. Note that the SYNC samples used
to estimate $\mathring{g}[n]$ and $\mathring{\beta}$ are those
earlier than the first STS sample corresponding to the CIR tap
$l_*$ to be validated. Also note that the SFD in Fig. \ref{fig:stsPktFormats}
provides a time gap of more than $2.9$ $\mu$s between SYNC
and STS as specified in 4z \cite{802.15.4z}. Algorithm \ref{algoSecRx}
makes no assumption regarding whether the SYNC and SFD are under attack
or not. We can establish the security of the detector in
(\ref{validation}) as follows. 

\begin{pro}\label{PropSecurityPerf}	
When validating a CIR tap $l_*$ that is earlier than the true
physical first path, the false acceptance rate of the STS receiver
in Algorithm \ref{algoSecRx} is upper bounded as
\begin{equation}\label{FA}
{\sf Pr}\left({\sf Accept~}l_*\Big|l_*<{\tau_0}/{T_0}\right)
\le \exp\left(-\gamma^2/2\right),
\end{equation}
where the probability is with respect to the random STS
sequence under Assumption \ref{Assumption1}.
Meanwhile, the upper bound in (\ref{FA}) is valid under
arbitrary feasible attacks.
\end{pro}

\begin{proof}
	Let $T_n := \sum_{k=0}^{n-1} x[k]s[k]$. We have
	\begin{equation}\label{eq1}
		\begin{split}
			&{\sf E}[T_{n+1}|T_{n}, T_{n-1},...,T_0] \\
			& = {\sf E}[T_{n}+x[n]s[n]|T_{n}, T_{n-1},...,T_0] \\
			& = T_{n} + {\sf E}[x[n]s[n]|T_{n}, T_{n-1},...,T_0] = T_{n},
		\end{split}
	\end{equation}
	where we have leveraged the independence between $x[n]$ and $s[n]$
	under arbitrary feasible attacks in obtaining the second equality. 
	From (\ref{eq1}), we see the sequence $\{T_n\}$ is
	a martingale with bounded increments, i.e.,
	\begin{equation}\label{eq2}
		|T_n-T_{n-1}|\le 1, \forall n. 
	\end{equation}
	According to the Azuma-Hoeffding inequality
	\cite{1963Hoeffding,RandomProcbook}
	as summarized in Appendix \ref{AppendixHoeffding}, $\forall \epsilon>0$,
	we have
	\begin{equation}\label{eq3}
		{\sf Pr}(T_Q\ge \epsilon) \le \exp\left(-{\epsilon^2}\Big/{2Q}\right).
	\end{equation}
	The false acceptance rate with the detector in (\ref{validation})
	is thus given by
	\begin{equation}\label{eq4}
		{\sf Pr}(T({\bm x})\ge \gamma) = {\sf Pr}(T_Q \ge \sqrt{Q}\cdot \gamma)
		\le \exp\left(-{\gamma^2}\big/{2}\right).
	\end{equation}
\end{proof}

The upper bound in (\ref{FA}) builds the foundation for the
security of the STS receiver in Algorithm \ref{algoSecRx}
according to Definition \ref{SRdef}. To see this, let $\eta$
denote the estimate of the first path timing $\tau_0$ from
the CIR $\mathring{g}[n]$. Various schemes are available
to estimate $\tau_0$ by taking into account the underlying
structure of the CIR as in (\ref{aggPulse}), i.e.,
$g(t) = \sum_{l\ge 0} \alpha_l p_R(t-\tau_l)$. Let $\chi$ 
denote the peak timing of the aggregate receive pulse, i.e.,
$\chi:=\arg\max_{t}|p_R(t)|$. From
Proposition \ref{PropSecurityPerf}, we can readily
establish the following corollary.

\begin{coro}\label{CoroSecurityPerf}	
For a given estimate of the first path timing $\eta$
and a prescribed upper bound on the false acceptance rate
$\rho$, we set $l_* = \lfloor (\eta+\chi)/T_0\rfloor$ and
$\gamma = \sqrt{2\ln(1/\rho)}$ in Algorithm \ref{algoSecRx}.
Meanwhile, this timing estimate $\eta$ is accepted only when
the tap $l_*$ is validated as in (\ref{validation}).
Accordingly, the probability of false acceptance can be
bounded as
\begin{equation}\label{SecTimingFA}
{\sf Pr}({\sf Accept}~ \eta|\eta <\tau_0 - \chi) 
= {\sf Pr}({\sf Accept}~ \l_*|l_* T_0 < \tau_0)
\le \rho,
\end{equation}
where the probability is with respect to the random STS
sequence under Assumption \ref{Assumption1}.
\end{coro}

According to Definition \ref{SRdef}, the upper bounds on the
false acceptance rate in (\ref{FA}) and (\ref{SecTimingFA})
under arbitrary feasible attacks imply that the STS receiver
in Algorithm \ref{algoSecRx} indeed performs secure ranging
with an implementation headroom determined by the pulse
shape, i.e., $\Delta=\arg\max_{t}|p_R(t)|$.

\subsubsection{Uniqueness of UWB-IR}
The fundamental property that enables UWB-IR as an ideal candidate
to realize secure ranging is that the pulse width is in the order
of nanoseconds. This property makes it possible
to realize an implementation headroom as small as several nanoseconds.
In particular, for the precursor-free impulse recommended by IEEE
in \cite{802.15.4z}, the time distance from the starting to the
peak is about $2.5$ ns, i.e., $\arg\max_{t}|p_T(t)|\approx 2.5$ ns.
Meanwhile, the receiver filter impulse response $p_A(t)$ could also
be designed carefully such that the implementation headroom is
minimized. For example, we can limit the headroom to $2.5$ ns as
well, i.e.,
$\chi=\arg\max_{t}|p_R(t)|=\arg\max_{t}|p_T(t)\ast p_A(t)|\approx 2.5$ ns.

\subsubsection{Special Case with $\mathring{\beta}=\infty$}
We can look into a special case
about the saturation in (\ref{saturation}) when the scaling
factor $\mathring{\beta}$ is chosen to be $\infty$. In this
case, the saturated signal in (\ref{saturation}) is either
$+1$ or $-1$. As the tap $l_*$ is earlier than the first path,
it can be seen that the decision metric in (\ref{decMetric})
is a summation of IID zero-mean binary random variables.
The central limit theorem \cite{RandomProcbook} indicates
that the false acceptance rate can be approximated as
\begin{equation}\label{FA2}
	{\sf Pr}({\sf Accept~}l_*|l_*<\tau_0/T_0)
	= {\sf Pr}(T({\bm x})\ge \gamma) \approx  {\cal Q}(\gamma),
\end{equation}
where ${\cal Q}(\cdot)$ denotes the tail distribution function of the
standard normal distribution. It is interesting to recognize that the
bound in (\ref{FA}) is indeed the Chernoff bound on ${\cal Q}(\gamma)$,
i.e., ${\cal Q}(\gamma)\le \exp\left(-\gamma^2/2\right)$. This further
illustrates the general applicability of the result in Proposition
\ref{PropSecurityPerf}.

\subsection{Security under Clock Offset Attack}\label{sSec:ClockOffsetAttk}

Although the trusted ranging devices comply with the $\pm 20$
ppm clock accuracy requirement, an attacker does not need to follow.
In fact, an attacker could manipulate its clock offset and
try to advance the RMARKER arrival timing. Next we will show
how an attacker can advance the RMARKER at the RXD by exploiting a
slower clock as in \cite{2023Usenixsecurity}. 
We will further explain why this is not posing a threat
to the security of the STS receiver in Algorithm \ref{algoSecRx}.

\begin{figure}[t]
	\centering
	\epsfig{file=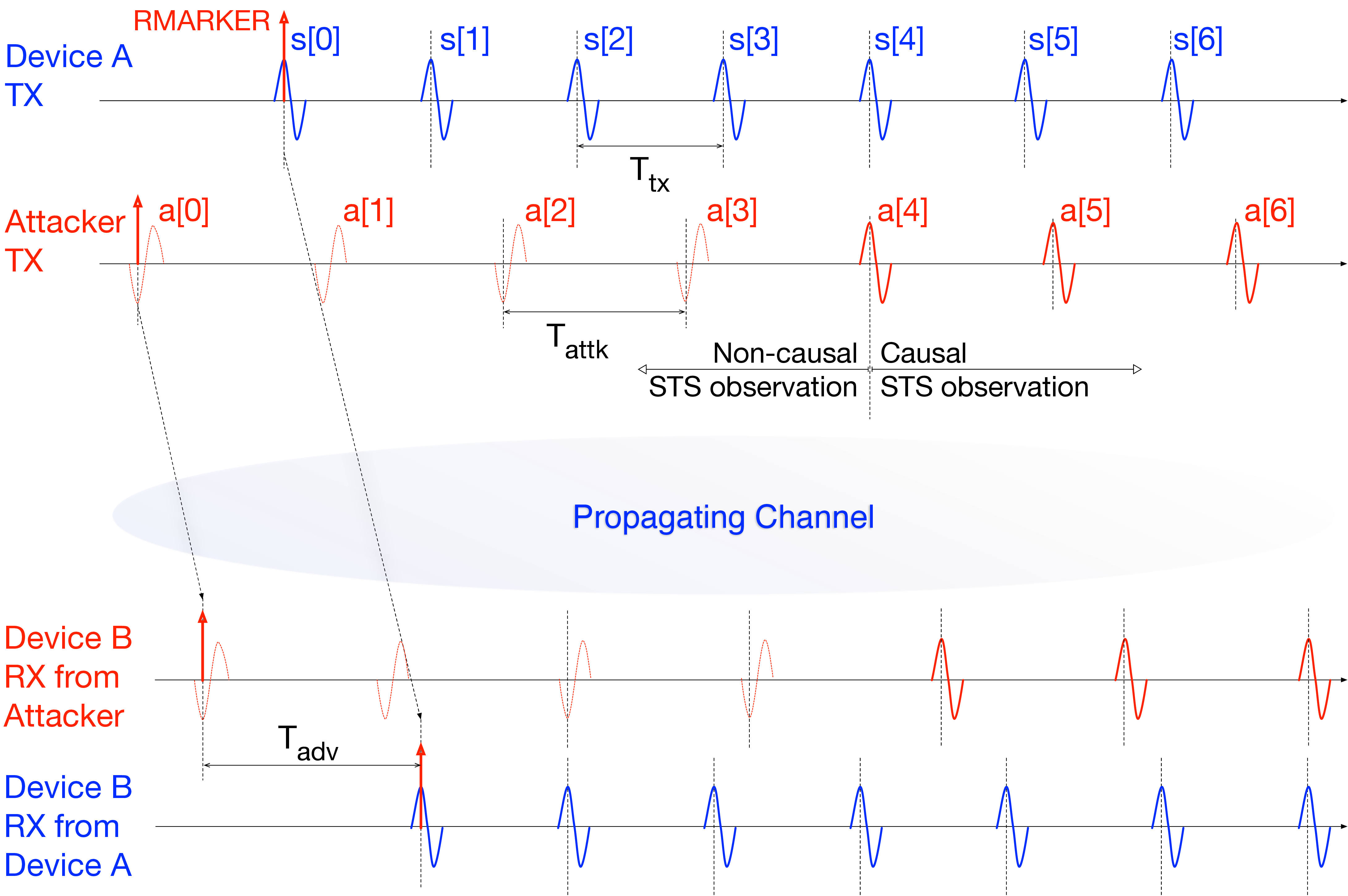, width=0.95\linewidth}
	\caption{Attack via a slower clock. The inter-pulse spacing from
		the attacker is larger than that from Device A, i.e., 
		$T_{\sf attk}> T_{\sf tx}$. The attacker can advance
		the RMARKER location at Device B by transmitting the attack
		sequence earlier than the actual STS sequence from Device A.
		Due to the slower clock at the attacker, the beginning portion
		of the attack sequence is independent to the actual STS while
		the later portion can follow the STS correctly. Specifically,
		when $k\ge 4$, the attacker can follow the STS by listening
		first and then transmitting, i.e., $a(k)=s(k)$, $k\ge 4$.
		In this way, the arrival time of the RMARKER is advanced
		and the corresponding CIR tap could be validated due to the
		fact that the attack sequence matches the STS partially.}  
	\label{fig:SlowClockAttk}
\end{figure}

When the attacker utilizes the same clock frequency as the TXD, our
previous analyses demonstrate that any CIR tap earlier than the true
first path will be rejected subject to a prescribed false acceptance
rate. In other words, the attacker cannot advance the arrival timing
of the RMARKER at the RXD.

On the other hand, the attacker can exploit a clock slower than the
TXD to advance the RMARKER at the RXD. One example of ``Slow Clock Attack''
is shown in Fig. \ref{fig:SlowClockAttk}. Let the inter-pulse spacing
be $T_{\sf attk}$ at the attacker and $T_{\sf tx}$ at the TXD.
For a length-$Q$ STS sequence, we can define $T_{\sf adv, max}$ as
\begin{equation}
T_{\sf adv, max} := Q(T_{\sf attk}- T_{\sf tx}).
\end{equation}
In order for the
attacker to ensure a portion of the attack sequence matches the actual
STS from the legitimate TXD, the amount of timing advance of the
RMARKER will have to be less than $T_{\sf adv, max}$. Note that the
RXD will sync to the attacker's clock and sample the received waveform
in (\ref{rxSig3}) at the clock rate of the attacker when the attacker's
clock ticks differently from that of the TXD. More remarks on
the sampling of the received waveform at the RXD are in Appendix
\ref{AppendixRXDSampling}.

The nominal inter-pulse spacing is $T_{\sf nom} = L T_c$ with
$L$ being the spreading factor. Consider the $\pm 20$ ppm
requirement on the clock accuracy at both the TXD and the
RXD. Meanwhile, we assume that the RXD will reject any input
signal that exhibits a relative clock offset beyond $\Gamma$
ppm with respect to its local clock. Then we can get
$T_{\sf attk}- T_{\sf tx} \le L T_c \cdot (40+\Gamma)/10^{6}$.
Note that we typically have
$\Gamma = 40$ in practical implementations to allow
maximum interoperation between IEEE devices.
Within an STS segment of length $K_{\sf sts}$ in the unit of
$512T_c$, there are $Q = 512\cdot K_{\sf sts}/L$ UWB pulses.
Accordingly, we have
\begin{equation}\label{maxAdv}
T_{\sf adv, max} < K_{\sf sts} \cdot 512 {T_c} \cdot (40+\Gamma)/ 10^{6}.
\end{equation}
We see $T_{\sf adv, max}$ is proportional to $K_{\sf sts}$.
A smaller $K_{\sf sts}$ leads to a smaller $T_{\sf adv, max}$.
When $K_{\sf sts} = 32$ and $\Gamma=40$, we have
$T_{\sf adv, max} < 2.6$ ns.

By allowing extra implementation headroom in (\ref{SecTimingFA}),
we can establish the security of Algorithm \ref{algoSecRx} in the
presence of clock offsets similar to Corollary \ref{CoroSecurityPerf}
as follows.

\begin{coro}\label{CoroClkOffset}	
For a given estimate of the first path timing $\eta$
and a prescribed upper bound on the false acceptance rate
$\rho$, we set $l_* = \lfloor (\eta+\chi)/T_0\rfloor$ and
$\gamma = \sqrt{2\ln(1/\rho)}$ in Algorithm \ref{algoSecRx}.
Meanwhile, this timing estimate $\eta$ is accepted only when
the tap $l_*$ is validated as in (\ref{validation}).
Accordingly, the probability of false acceptance can be
bounded as
\begin{equation}
\begin{split}
&{\sf Pr}({\sf Accept}~ \eta|\eta <\tau_0 - \chi-T_{\sf adv, max})\\
& = {\sf Pr}({\sf Accept}~ \l_*|l_* T_0 < \tau_0-T_{\sf adv, max})
\le \rho,
\end{split}
\end{equation}
where $\chi:=\arg\max_{t}|p_R(t)|$,
$T_{\sf adv, max} := 
K_{\sf sts} \cdot 512 {T_c} \cdot (40+\Gamma)/ 10^{6}$,
and the probability is with respect to the random STS
sequence under Assumption \ref{Assumption1}.
\end{coro}

Considering the results in Corollary \ref{CoroSecurityPerf} and
Corollary \ref{CoroClkOffset}, for the
sake of clarity in exposition, we will assume the clock at the
attacker is ideal as both the TXD and the RXD in this paper
unless otherwise specified.

\subsection{Detection Performance}\label{sSec_DetPerf}

We next check how well the STS receiver design in Algorithm
\ref{algoSecRx} performs when the CIR tap $l_*$ corresponds to the
true first path in (\ref{yn}). After the interference cancellation
step in Algorithm \ref{algoSecRx}, we have
\begin{eqnarray}
\tilde{x}[n] = g[l_*]s[n] + \tilde{w}[n],
\label{H1}
\end{eqnarray}
where $\tilde{w}[n]$ models both the additive noise and
the residual interference from other taps due to non-ideal
cancellation. 

After the saturation step in (\ref{saturation}), we have
\begin{equation}\label{saturatedSig}
\begin{split}
x[n] 
& = {\mathbb Q}\left({\sf Re}\left\{
	\mathring{\beta} \cdot \mathring{g}[l_*]^{\dagger}\cdot \tilde{x}[n]
	\right\}
	\right)\\
&= {\mathbb Q}\left({\sf Re}\left\{
\mathring{\beta} \cdot \mathring{g}[l_*]^{\dagger} (g[l_*] s[n]+ \tilde{w}[n])
\right\}
\right)\\
&= {\mathbb Q}\left( \bar{h} s[n] + \bar{w}[n] \right)
\end{split},
\end{equation}
where notation $(\cdot)^{\dagger}$ denotes the conjugate operation,
$\bar{h} := {\sf Re}\{\mathring{\beta} \cdot 
\mathring{g}[l_*]^{\dagger} g[l_*] \}$, and 
$\bar{w}[n] := {\sf Re}\{\mathring{\beta} \cdot 
\mathring{g}[l_*]^{\dagger} \tilde{w}[n]\}$.
By defining ${\mathbb M}(s[n]):= {\sf E}[x[n]|s[n]]$, we have
\begin{equation}
x[n] = {\mathbb M}(s[n]) + e[n], \label{SigDecompose}
\end{equation}
where ${\mathbb M}(s[n])$ is the conditional mean estimate of $x[n]$
with $s[n]$ and $e[n]:=x[n] - {\mathbb M}(s[n])$ stands for the
estimation error with a zero mean, i.e.,
\begin{equation}
{\sf E}[e[n]] = {\sf E}[x[n]] - {\sf E}\big[{\sf E}[x[n]\big|s[n]]\big] = 0. 
\end{equation}
Further, it can be shown that $s[n]$ and  $e[n]$ is uncorrelated, i.e.,
\begin{equation}\label{uncorrelatedError}
\begin{split}
&{\sf E}[s[n]e[n]]
={\sf E}\Big[ 
{\sf E}\big[ s[n] \left(x[n]- {\sf E}\big[x[n]\big|s[n]\big]\right) 
\big| s[n] \big] \Big]\\
&={\sf E}\Big[ s[n]{\sf E}\big[x[n]\big| s[n]\big] 
             - s[n]{\sf E}\big[x[n]\big|s[n]\big] \Big] = 0
\end{split}.
\end{equation}

From the decomposition in (\ref{SigDecompose}), we can establish
the following proposition to characterize the detection performance
of the STS receiver in Algorithm \ref{algoSecRx}.

\begin{pro}\label{PropDetperf}
Assume the estimation error $e[n]$ in (\ref{SigDecompose}) is
uncorrelated with $s[n]$ conditioned on the past samples
$\{x[k]\}_{k<n}$ and the past STS $\{s[k]\}_{k<n}$. When validating
a CIR tap $l_*$ corresponding to the true first path, the miss rate
of the reference receiver in Algorithm \ref{algoSecRx} can be upper
bounded as
\begin{equation}\label{MissProb}
\begin{split}
{\sf Pr}({\sf Reject ~}l_*) &= {\sf Pr}(T({\bm x})< \gamma) \\
&\le \exp\big(
-{Q\big(\bar{C}-{\gamma}\big/{\sqrt{Q}}\big)^2}\big/{2}
\big),
\end{split}
\end{equation}
where 
$C_n:= {\sf E}[{\mathbb M}(s[n]) s[n]]$,
$\bar{C} := \sum_{n=0}^{Q-1}C_n/Q$, and
$\bar{C}\sqrt{Q}$ is large enough such that $\bar{C}>\gamma/\sqrt{Q}$.
The probability in (\ref{MissProb}) is with respect to the random STS
under Assumption \ref{Assumption1}.
\end{pro}

Appendix \ref{AppendixDetPerf} shows the proof outline.
Note the condition $\bar{C}>\gamma/\sqrt{Q}$ can be easily met when
the STS length $Q$ gets large and ${\sf E}[{\mathbb M}(s[n])s[n]]>0$.
The upper bound in (\ref{MissProb}) indicates the STS receiver
in Algorithm \ref{algoSecRx} achieves optimal asymptotic detection
performance in the sense that the miss rate decreases to zero
exponentially with respect to the STS length. Also note that the
correlation value $C_n= {\sf E}[{\mathbb M}(s[n]) s[n]]$ can be maximized
by optimizing the scaling parameter $\mathring{\beta}$ in Algorithm
\ref{algoSecRx}.

As shown in (\ref{uncorrelatedError}), the estimation error $e[n]$ is
uncorrelated with $s[n]$. Besides, the estimation error $e[n]$ and the
current STS sample $s[n]$ become uncorrelated conditioned on
$\{x[k]\}_{k<n}$ and $\{s[k]\}_{k<n}$ when the residual term
$\tilde{w}[n]$ in (\ref{H1}) is independent of the past observations
and STS. This is the case when the interference cancellation
in Algorithm \ref{algoSecRx} is perfect. The detection performance
of one practical implementation with non-ideal interference
cancellation will be characterized numerically in Section \ref{SecSim}.

\subsection{More Discussions}\label{sSec_moreDiscuss}

Proposition \ref{PropSecurityPerf} ensures the security of the reference
design in Algorithm \ref{algoSecRx} and Proposition \ref{PropDetperf}
guarantees the detection performance. Compared to the STS receiver
based on the simple idea of CIR tap thresholding as reviewed in Section
\ref{SecCIRThr}, the key factors enabling the security of the reference
design can be highlighted as follows.
\begin{itemize}
\item When canceling the interference from other CIR taps, it is
important to rely on the tap estimates that only utilize the signal
arriving earlier than the STS. This is why Algorithm
\ref{algoSecRx} employs the estimated CIR from the SYNC in
the packet instead of the STS itself. On the other hand, when
canceling the interference on a particular tap with the CIR
estimated from the STS, an adversary could exploit this fact and
attack the ranging as illustrated in Algorithm \ref{algoAdapAttk};

\item When canceling the interference from other CIR taps, it is
important to only cancel the taps later than the tap being validated.
This is why we only cancel the taps at $l_*+\zeta M$ with $\zeta>0$
in (\ref{intfCancel}). Meanwhile, it is worthwhile to point out that
the security of the reference design in Algorithm \ref{algoSecRx}
does not necessitate cancellation of all the interfering taps.
However, a cleaner signal after cancellation in (\ref{intfCancel})
will lead to improved detection performance;
	
\item The saturation operation in (\ref{saturation}) is also critical
in establishing the upper bound on the false acceptance rate as in
(\ref{FA}) under arbitrary feasible attacks.
The contribution from each individual signal sample towards the decision
metric in (\ref{decMetric}) is also restricted by this saturation.
Thus the attacker can not achieve extra gains by putting more emphasis
on particular signal samples.

\end{itemize}

In Algorithm \ref{algoSecRx}, we have utilized a particular saturation
function as defined in (\ref{funcQ}), which allows efficient hardware
implementation. Note that a different saturation function, e.g., the
logistic sigmoid \cite{MachineLearningbook}, can be employed in the
step of ``Saturation'' without undermining the security of the STS
receiver in Algorithm \ref{algoSecRx}.

Also note that the discrete samples in (\ref{rxSig3})
are actually the outputs of the analog-to-digital converter (ADC) at
the receiver. Since the ADC outputs are confined within an interval,
the post-cancellation signal in (\ref{intfCancel}) is thus already
bounded even without the saturation function ${\mathbb Q}(\cdot)$ in
(\ref{funcQ}). Accordingly, the security and the upper bound on the
false acceptance rate can be established similarly as in Proposition
\ref{PropSecurityPerf}. Nevertheless, the ``Saturation'' step
in Algorithm \ref{algoSecRx} enables optimization of the detection
performance.

\section{Numerical Performance Evaluations}\label{SecSim}

In this section, we corroborate the designs presented in previous
sections numerically. In particular, we simulate one HPRF system
as illustrated in Fig. \ref{fig:sysModel} with the following
configurations:
\begin{itemize}
\item STS packet format: STS Packet Config 3 in Fig. \ref{fig:stsPktFormats};
	
\item STS configuration: one STS segment of length ${K_{\sf sts}}=64$ with
spreading factor: $L=4$. The STS sequence length is $Q=8192$;

\item Sampling rate: $T_0=T_c/2$ in (\ref{rxSig3}), which corresponds
to twice of the chip rate in 4z HRP;

\item  Receive pulse $p_R(t)$ in (\ref{aggPulse}): modeled as the precursor-free
pulse plotted in Fig. \ref{fig:minPhaseSRRC} in Appendix \ref{AppendixPCF};

\item CIR from legitimate TXD to RXD: the channel between
legitimate TXD and RXD is modeled as one multi-path channel with
the first path corresponding to tap $126$ of the length-$256$
CIR $g[n]$ in (\ref{rxSig3}), i.e., 
$g[n] = \alpha_0 p_R[n-\bar{\tau}_{0}] + 
\sum_{l=1}^{n_{\sf Intf}} \alpha_l p_R[n-\bar{\tau}_l]$,
where $\alpha_0=1$, $\bar{\tau}_{0} = 126$, and $n_{\sf Intf}$
denotes the number of interfering taps;

\item Additive noise: the noise $w[n]$ in (\ref{rxSig3})
is modeled as AWGN with the variance dictating the
signal-to-noise (SNR) ratio at the RXD;

\item CIR from attacker to RXD: the channel between the
attacker and RXD is modeled as a single tap with $\theta=1$
and $\bar{\tau}_{0}=126$ in (\ref{rxSigAttk2});

\item Timing and frequency synchronization: we assume perfect
timing and frequency synchronization at the TXD, the RXD, and
the attacker;

\item Scaling parameter in the reference STS receiver: we will
assume the scaling parameter $\mathring{\beta}$ in Algorithm
\ref{algoSecRx} is set to be $\infty$. The saturated signal in
(\ref{saturation}) is thus simply the sign of the input argument.
\end{itemize}

\subsection{Effectiveness of Attack in Algorithm \ref{algoAdapAttk}}
\label{SubSec_SimAdapAttk}

First, we examine the effectiveness of the adaptive attack scheme
in Algorithm \ref{algoAdapAttk}. We neglect the additive noise and
assume the legitimate signal is overwhelmed by the attack one.
Accordingly, the received signal at the RXD only consists of the
adaptive attack signal as in (\ref{rxSigAttk2}). The RXD recovers
an CIR estimate of length $256$ from the received attack
signal according to (\ref{slcCIR}) in Proposition \ref{PropCIREst}.
Fig. \ref{fig:Lambda2H1} and \ref{fig:Lambda2H15}
depict the statistics about the estimated CIRs over $1000$ random
STS realizations under different settings of $\Lambda$ and $H$ in
Algorithm \ref{algoAdapAttk}. From the statistics, we can see that
the adaptive attack is indeed very effective in creating a peak
tap at $112$. Since $\Lambda=2$ and $M=8$, this peak tap is
$\Lambda M = 16$ taps before the peak tap corresponding to the actual
first path at $128$, which is as predicted in Proposition
\ref{PropAdapAttk}. Meanwhile, we can also see that the amplitude
of the early tap at $112$ indeed gets larger as we set $H$ to a
higher value in Algorithm \ref{algoAdapAttk}.

\begin{figure}[t]
	\centering
	\epsfig{file=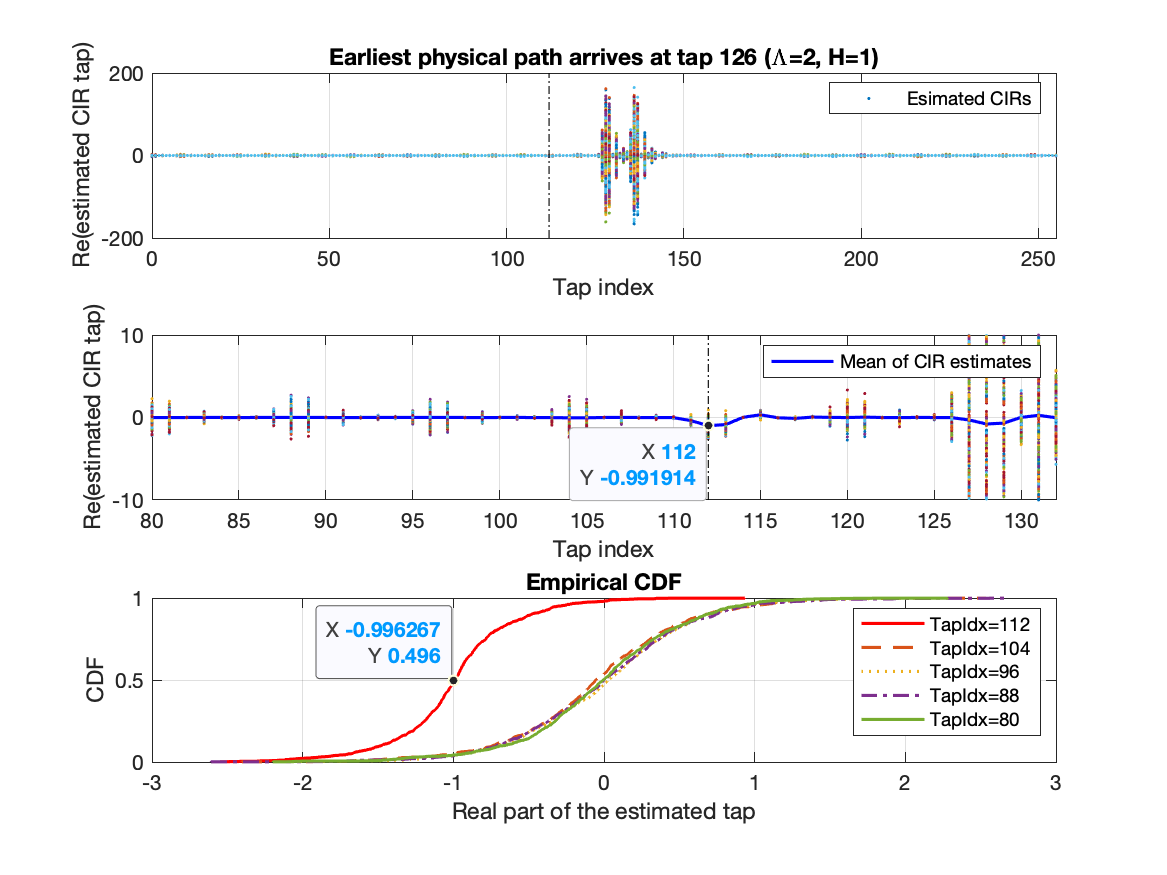, width=0.88\linewidth}
	\caption{Statistics of CIR estimates under adaptive attacks with
		$\Lambda=2$ and $H=1$.}
	\label{fig:Lambda2H1}
\end{figure}

\begin{figure}[t]
	\centering
	\epsfig{file=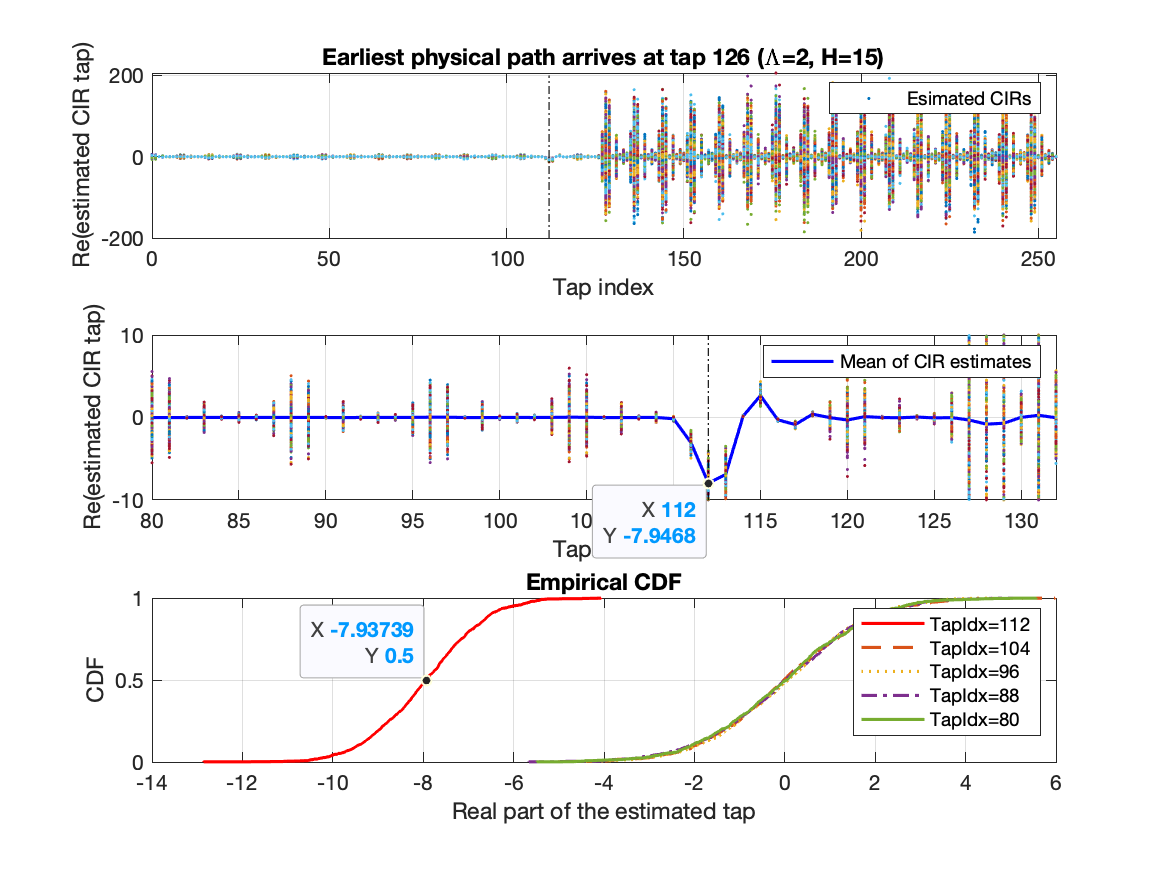, width=0.88\linewidth}
	\caption{Statistics of CIR estimates under adaptive attacks with
		$\Lambda=2$ and $H=15$.}
	\label{fig:Lambda2H15}
\end{figure}

\begin{figure}[t]
	\centering
	\epsfig{file=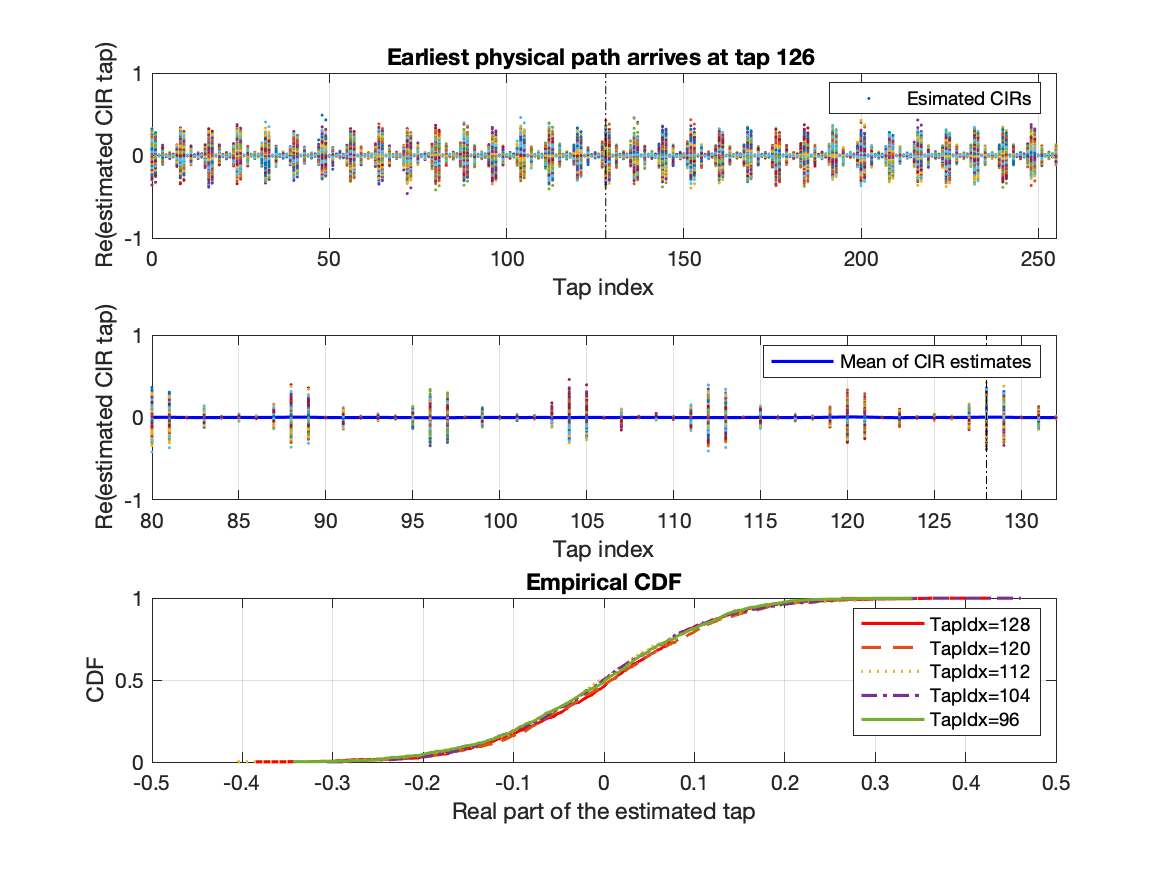, width=0.88\linewidth}
	\caption{Statistics of CIR estimates under ``ghost peak'' attacks
		in \cite{2022GhostPeak}.}
	\label{fig:GhostAttk}
\end{figure}

In comparison, Fig. \ref{fig:GhostAttk} shows the statistics of the
CIR estimates when the attacker performs the ``ghost peak'' attacks
as in \cite{2022GhostPeak} by transmitting independent random binary
sequences. Specifically, the sequence $\{a[k]\}$ in (\ref{rxSigAttk})
is a white Bernoulli process with each sample taking the value of $+10$
or $-10$ with equal chances. Fig. \ref{fig:GhostAttk} shows that
the ``ghost peak'' attacks can not ensure an early CIR peak at
a prescribed location.

From the above presented results, ones can see that the adaptive
attack scheme in Algorithm \ref{algoAdapAttk} is effective in
generating a fake peak at a specific location in the CIR. On the
one hand, it implies those STS receivers that simply compare the
amplitudes of the CIR taps against certain thresholds are prone to
such attacks. Meanwhile, this also suggests there is no security
guarantee to validate one particular CIR tap by comparing
the CIR estimate from the SYNC in the packet with the one
recovered from the STS portion.

\subsection{Security of STS Receiver in Algorithm \ref{algoSecRx}}

Next, we examine the security performance of the reference STS receiver
in Algorithm \ref{algoSecRx}. In the absence of attacks, the RXD only
receives the legitimate transmission from the TXD as in (\ref{rxSig3}).
Fig. \ref{fig:CIREsts_DiffSNR} shows instances of the estimated
CIRs from the received STS signals under different SNR levels.
Further, Fig. \ref{fig:SR_EarlyPath_DiffSNR} shows the statistics of
the decision metric in (\ref{decMetric}) when validating CIR taps
earlier than the physical first path in the absence of attacks. From the
results, we see that the decision metric follows the Gaussian distribution
with zero mean and unit variance no matter how strong the noise gets.
The desired false acceptance rate of $\rho$ can be realized by
configuring the detection threshold in (\ref{validation}) as 
$\gamma = \mathcal{Q}^{-1}(\rho)$.

In the case of the adaptive attack in Algorithm \ref{algoAdapAttk},
we also employ the reference STS receiver in Algorithm
\ref{algoSecRx} to validate CIR taps earlier than the true
physical path. In particular, by setting $\Lambda=2$ and
$H=15$ as in Fig. \ref{fig:Lambda2H15}, the statistics of
the decision metric in (\ref{decMetric}) in validating
an early CIR tap are shown in Fig. \ref{fig:SR_FakePath_DiffSNR}.
Similarly, Fig. \ref{fig:SR_FakePath_DiffSNR} also shows the statistics
of the decision metric in (\ref{decMetric}) in the case of
``ghost peak'' attacks.

From the above results, we can see that the decision metric
in (\ref{decMetric}) follows the standard normal distribution
asymptotically whenever validating non-existing fake taps.
Accordingly, the false acceptance probability is always
upper bounded as in Proposition \ref{PropSecurityPerf}.

\begin{figure}[t]
\centering
\epsfig{file=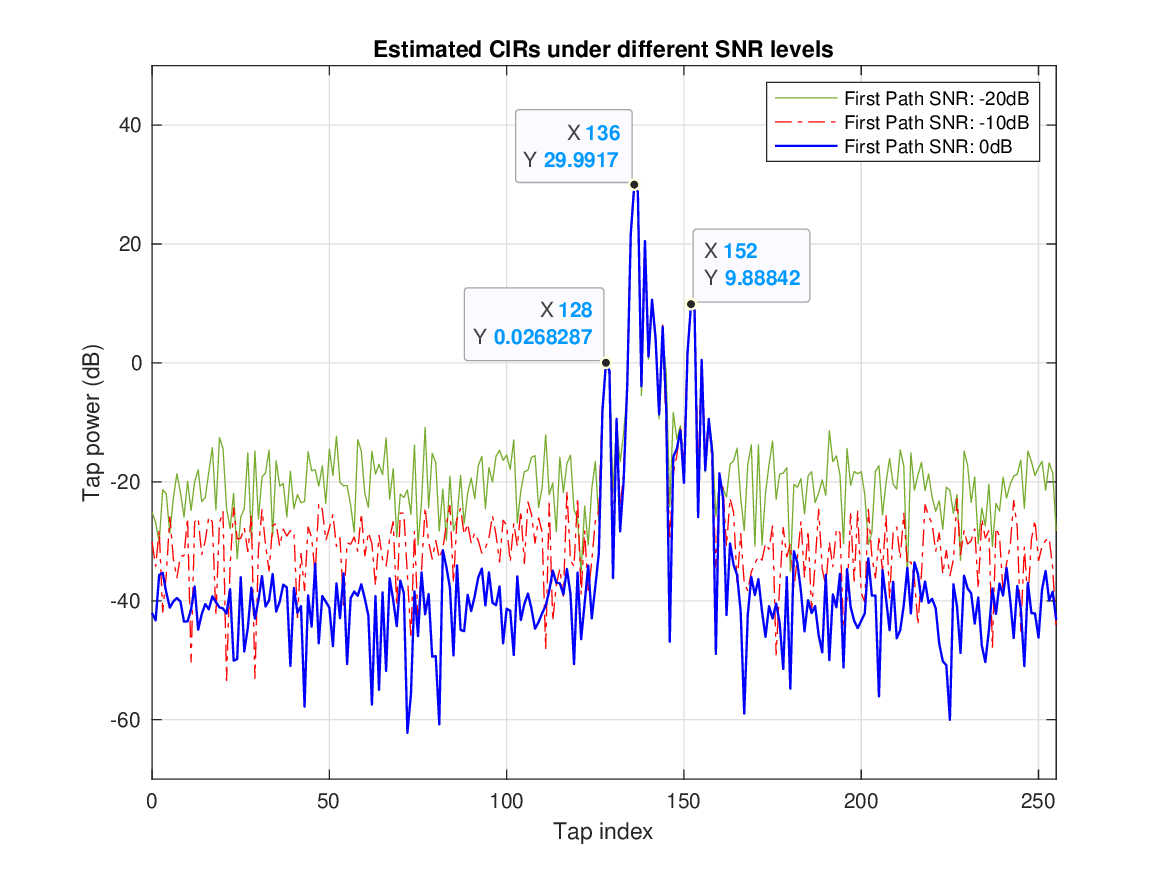, width=0.88\linewidth}
\caption{CIR estimates from the STS signals under different
SNRs without attacks ($n_{\sf Intf}=3$,
$(\alpha_0, \bar{\tau}_0)=(0{\rm dB}, 126)$
$(\alpha_1, \bar{\tau}_1)=(30{\rm dB}, 134)$,
$(\alpha_2, \bar{\tau}_2)=(0{\rm dB}, 142)$,
$(\alpha_3, \bar{\tau}_3)=(10{\rm dB}, 150)$).}
\label{fig:CIREsts_DiffSNR}
\end{figure}

\begin{figure}[t]
	\centering
	\epsfig{file=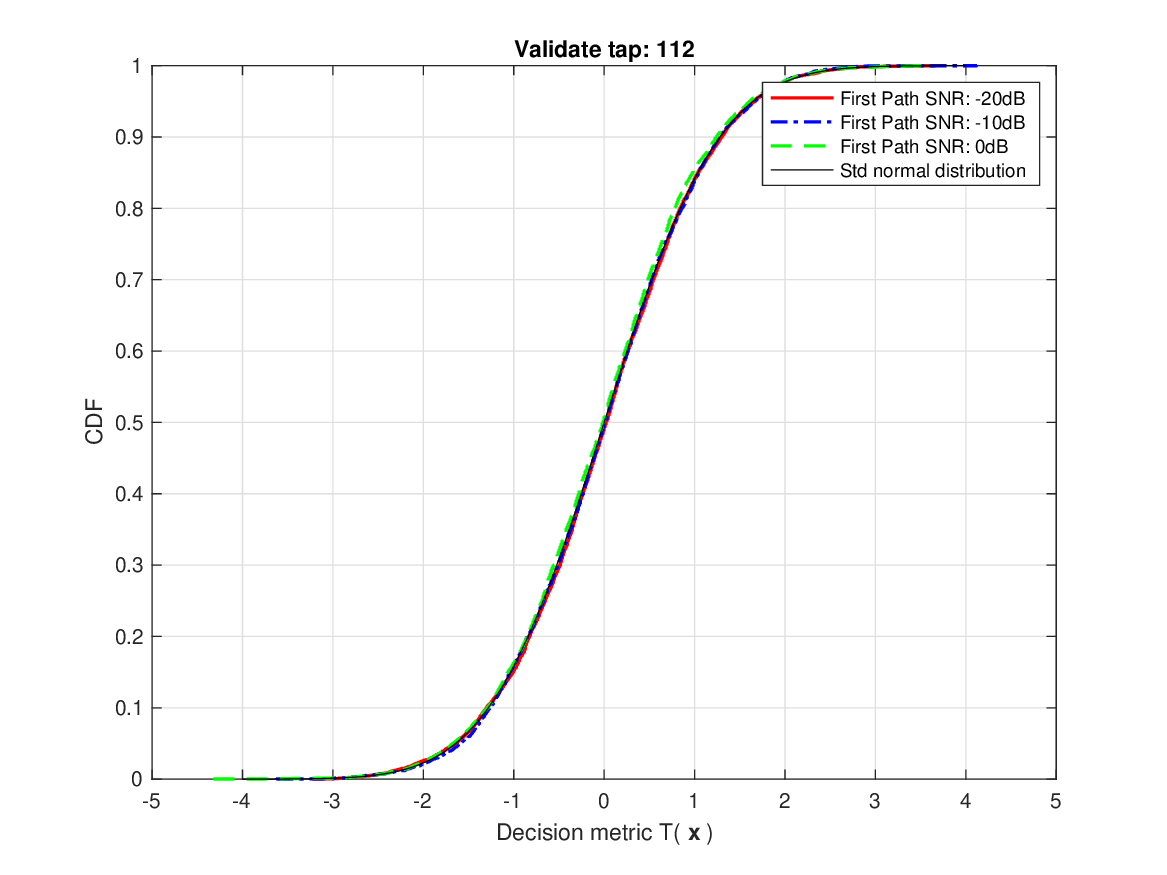, width=0.88\linewidth}
	\caption{CDFs of the decision metric when validating the
		CIR tap $112$ in the absence of attacks.}
	\label{fig:SR_EarlyPath_DiffSNR}
\end{figure}

\begin{figure}[t]
	\centering
	\epsfig{file=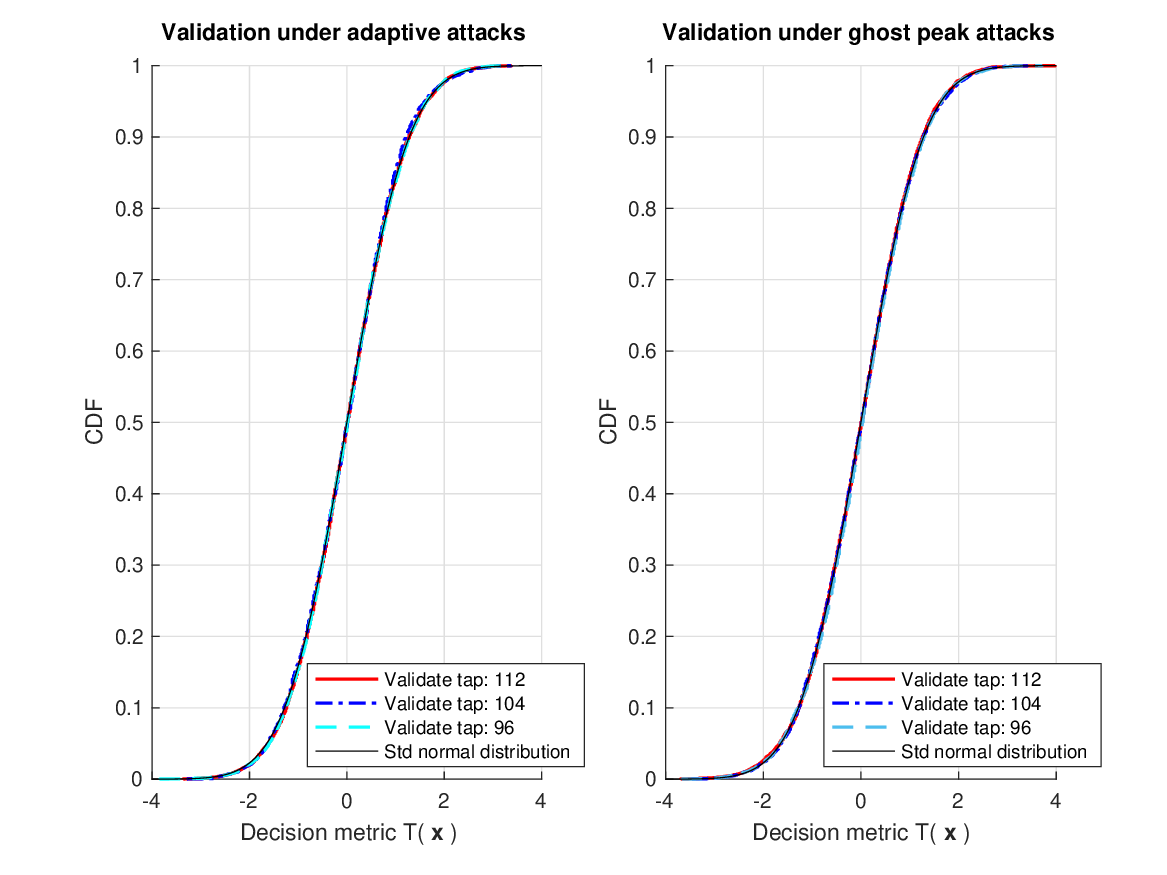, width=0.88\linewidth}
	\caption{CDFs of the decision metric when validating different
		CIR taps in the case of attacks. Left: adaptive attacks with
		$\Lambda=2$ and $H=15$. Right: ``ghost peak'' attacks.}
	\label{fig:SR_FakePath_DiffSNR}
\end{figure}

\subsection{Detection Performance of STS Receiver in Algorithm \ref{algoSecRx}}

In this section, we examine the detection performance of the reference
STS receiver in Algorithm \ref{algoSecRx} in the absence of attacks. 
In particular, the received signal at the RXD consists of the
legitimate STS waveform from the TXD and the AWGN as in (\ref{rxSig3}).
With the same CIR configuration as in Fig. \ref{fig:CIREsts_DiffSNR},
we next present the detection performance when validating the tap
of $l_*=128$, which corresponds to the pulse peak on the first path.
Note that we have $\arg\max_n |p_R[n]|=2$ and $\bar{\tau}_0=126$.
Fig. \ref{fig:SR_Det_Stat} shows the detection
performance of the first path when the detection threshold
in (\ref{validation}) is set to $\gamma = {\mathcal Q}^{-1}(10^{-6})$ to
ensure the false acceptance rate at $10^{-6}$. Note that the inter-tap
interference cancellation in (\ref{intfCancel}) may be non-ideal
either due to the errors in the estimates of the interference taps or
simply because we do not cancel all of the interference taps.
Fig. \ref{fig:SR_Det_Stat} also depicts the detection performance
with different amount of residual inter-tap interference. The results
indicate that the detection performance does not degrade much without
canceling those interference taps that are less than $10$dB with
respect to the first path.

Fig. \ref{fig:SR_Det_Stat} also shows the corresponding detection
performance when targeting different false acceptance rates, i.e.,
$\rho=10^{-6}$, $2^{-30}$, and $2^{-48}$.
To gain a higher level of security with
a smaller false acceptance rate, we need to set a higher detection
threshold in Algorithm \ref{algoSecRx}. The detection performance
thus degrades accordingly. Note that the target false rates of
$10^{-6}$ and $2^{-48}$ correspond to the medium and high security
levels in FiRa certification \cite{FiRaCert} respectively.
While ensuing the false acceptance rate of $2^{-48}$, Fig.
\ref{fig:SR_Det_Stat} shows that the reference STS receiver
in Algorithm \ref{algoSecRx} can detect the first path with a
probability more than $99\%$ as long as the SNR of the first
path is above $-20$dB.

\begin{figure}[t]
\centering
\epsfig{file=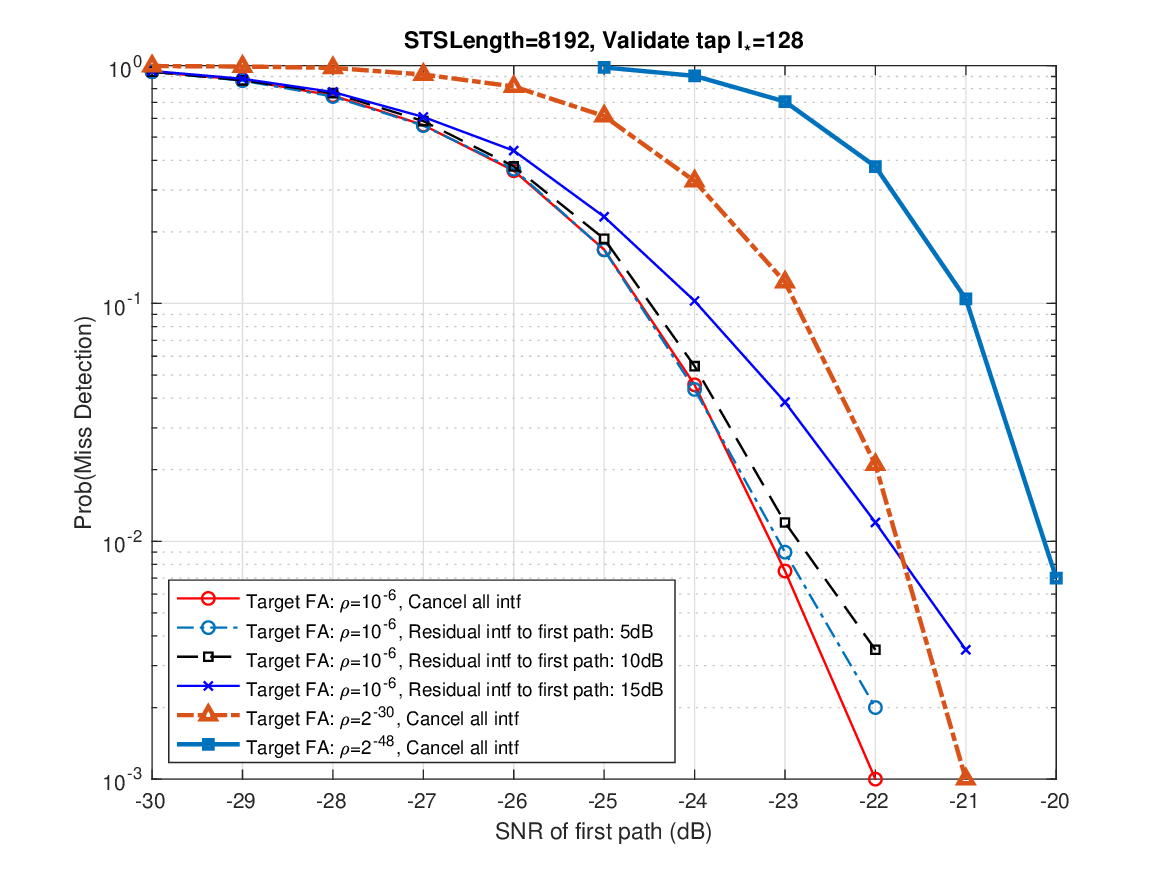, width=0.88\linewidth}
\caption{Miss detection performance with a target false acceptance rate
of $\rho = 10^{-6}, 2^{-30}, 2^{-48}$.
``Cancel all intf'': all the interference from
taps at $136$, $144$, and $152$ are canceled; 
``Residual intf to first path: $x$dB'': the amount of interference after
cancellation is $x$ dB over the first path power.}
\label{fig:SR_Det_Stat}
\end{figure}

\section{Conclusions}\label{SecConc}

In this paper, we have presented a reference STS receiver and demonstrated
that secure ranging can be achieved by employing the STS waveform in 4z
HRP. We have also characterized the performance bounds of the reference
secure STS receiver.

On one hand, the results in this paper address all the concerns about
the security of 4z HRP mode. On the other hand, the reference STS receiver
and the highlighted design principles in this paper will fortify the
foundation for all the use cases building on secure ranging with 4z
HRP UWB.

\bibliographystyle{IEEEtran}
\bibliography{SecRangingReflib}

\appendices
\section{Key Notations, Symbols, and Acronyms}\label{AppendixNotations}
All the key notations and symbols adopted in this paper are listed
in Table \ref{tab:SymbPaper} together with corresponding definitions.
The acronyms used in this paper are listed in Table \ref{tab:Acronyms}.

\begin{table}[htp!]
	\centering
	\caption{List of Symbols and Notations}
	\begin{tabular}{l|l}
		\hline
		\textbf{Symbol} & \textbf{Definition (Typical Values)}\\
		\hline
		$T_c$ & Chip duration ~($=1/499.2$ $\mu$s, $\approx 2$ ns)  \\
		$T_0$ & Sampling period ~($T_0 = T_c/\Omega$, $\Omega\ge 2$) \\
		$\{s[k]\}_{k=0}^{Q-1}$ & Length-$Q$ STS sequence ~($s[k]=\pm 1$, $Q=8192$) \\
		$L$ & Spreading factor ~($L\in\{4, 8\}$) \\
		$s(t)$ & Transmitted STS waveform in the baseband \\
		& from TXD \\
		$T_{\sf sts}$ & Length of one STS segment\\
		$K_{\sf sts}$ & Length of one STS segment in the unit of $512T_c$,\\
		& i.e., $T_{\sf sts} = K_{\sf sts}\times (512T_c)$ ~($K_{\sf sts}=64$)\\
		$p_T(t)$ & Transmit impulse in the baseband\\		
		$p_A(t)$ & Receiver analog front-end filter impulse response \\
		$p_R(t)$ & Aggregate receive pulse shape without channel\\
		$\chi$ & Peak timing of the receive pulse, i.e., \\
			   & $\chi = \arg\max_t |p_P(t)|$~ ($2.5$ ns)\\		
		$h(t)$ & Multi-path channel model\\
		$\alpha_l$ & The complex magnitude of the $l$-th path in $h(t)$ \\
		$\tau_l$ & The time delay of the $l$-th path in $h(t)$ \\
		$g(t); g[n] $ & Aggregate pulse shape with channel, a.k.a. CIR; \\
			          & $g[n]=g(nT_0)$\\
		$r(t); r[n]$ & Received signal at RXD; $r[n]=r(nT_0)$ \\
		$\eta$ & Estimate of the arrival timing of the first path \\
		$\Delta$ & Allowed implementation headroom in secure ranging \\
		$\rho$ & Prescribed upper bound on the false acceptance rate \\
		& ($\rho=2^{-48}$)\\
		$T_f;\hat{T}_f$ & ToF between ranging devices; the ToF estimate\\
		$\Lambda$ & A positive integer-valued step size of empirical \\
		& prediction in adaptive STS attack ($\Lambda\in {\mathbb{Z}^+}$)\\
		$H$ & Number of additional history samples used for \\
		& empirical prediction ($H\in {\mathbb{Z}^+} \cup \{0\}$)\\
		$\{a[k]\}_{k=0}^{Q-1}$ & Attack symbols in adaptive STS attack\\
		$a(t)$ & Waveform of adaptive STS attack\\
		$\theta$ & Path gain of the channel from the attacker to RXD\\
		$\bar{\tau}_l$ & The discrete value of $\tau_l$ when
		$\tau_l = \bar{\tau}_l T_0$\\
		$\{\mathring{g}[n]\}_{n=0}^{JM-1}$ & CIR estimate with the SYNC
		in the current\\
		& STS packet \\
		$\mathring{\beta}>0$ & Scaling parameter configured according to the\\
		& received signal power level over SYNC \\
		$l_*$ & Index of the tap to be validated in the CIR $\mathring{g}[n]$\\
		$T(\bm x)$ & Decision metric to validate a given CIR tap\\
		$\gamma$ & Detection threshold to accept a given CIR tap\\
		$\Gamma$ & Maximum relative clock offset in ppm allowed by \\
		&RXD ($\Gamma = 40$ ppm)\\
		${\cal Q}(\cdot)$ & Tail distribution function of the standard normal\\
		& distribution, a.k.a. Q-function\\
		${\sf Pr}(\cdot)$ & Probability of the specified event\\
		${\sf E}[\cdot]$ & Expectation of the specified random variable\\
		$ {\bm A}(i,j)$ & The element of matrix $\bm A$ at the intersection of \\
		& the $i$-th row and the $j$-th column\\
		$ {\bm A}(i,:)$ & The $i$-th row vector of matrix $\bm A$ \\
		$ {\bm A}(:,j)$ & The $j$-th column vector of matrix $\bm A$ \\	
		\hline
	\end{tabular}
	\label{tab:SymbPaper}
\end{table}

\begin{table}[htp!]
	\centering
	\caption{List of Acronyms}
	\begin{tabular}{l|l}
		\hline
		\textbf{Acronym} & \textbf{Full-form Meaning }\\
		\hline
		UWB-IR & ultra-wideband impulse radio \\
		ToF & time-of-flight \\
		LRP & low rate pulse repetition frequency \\
		HRP & high rate pulse repetition frequency \\
		PRF & pulse repetition frequency \\
		BPRF & base pulse repetition frequency \\
		HPRF & higher pulse repetition frequency \\
		PHY & physical layer \\
		SYNC & synchronization\\
		SFD & start-of-frame delimiter\\
		STS & scrambled timestamp sequence\\
		SF & spreading factor \\
		RMARKER & ranging marker\\
		TXD & transmit-device \\
		RXD & receive-device \\
		DS-TWR & double-sided two-way ranging \\
		CIR & channel impulse response \\
		\hline
	\end{tabular}
	\label{tab:Acronyms}
\end{table}

\section{UWB Standardization}\label{AppendixUWBStandard}
The Federal Communications Commission (FCC) in the United States authorized
the use of UWB for commercial purposes in 2002 \cite{FCC}. The first UWB
standard was proposed by the MultiBand OFDM Alliance (MBOA) in 2004 and
took a wideband orthogonal frequency-division multiplexing (OFDM) approach
\cite{MBOA}. Meanwhile, an IR-based proposal was put forward by the WiMedia
Alliance, which used a direct-sequence UWB (DS-UWB) approach \cite{DSUWB}.
IEEE 802.15.4 \cite{802.15.4} is a wireless networking standard that specifies
the physical layer (PHY) and medium access control (MAC) sublayers for a
low-rate wireless personal area network (LR-WPAN). In 2006, an amendment to
the standard, IEEE 802.15.4a, added a new PHY based on UWB-IR to provide
higher precision ranging and localization capability and higher aggregate
throughput.

In 2017, the IEEE 802.15 working group initiated efforts to define a series of
PHY enhancements to improve the capabilities of UWB devices. In 2018, the task
group 802.15.4z started to work on the corresponding amendment to IEEE 802.15.4.
This amendment, i.e., IEEE 802.15.4z \cite{802.15.4z}, was completed and
published in 2020. Amendment 4z enhanced the UWB PHYs and associated ranging
techniques. In particular, amendment 4z specified encrypted ranging waveforms
to increase the integrity and accuracy of the ranging measurements.
A more comprehensive overview of UWB standards and organizations is
available in \cite{2022UWBStdOverview}.

\section{STS Generation DRBG}\label{AppendixSTSDRBG}

\begin{figure}[h]
\centering
\epsfig{file=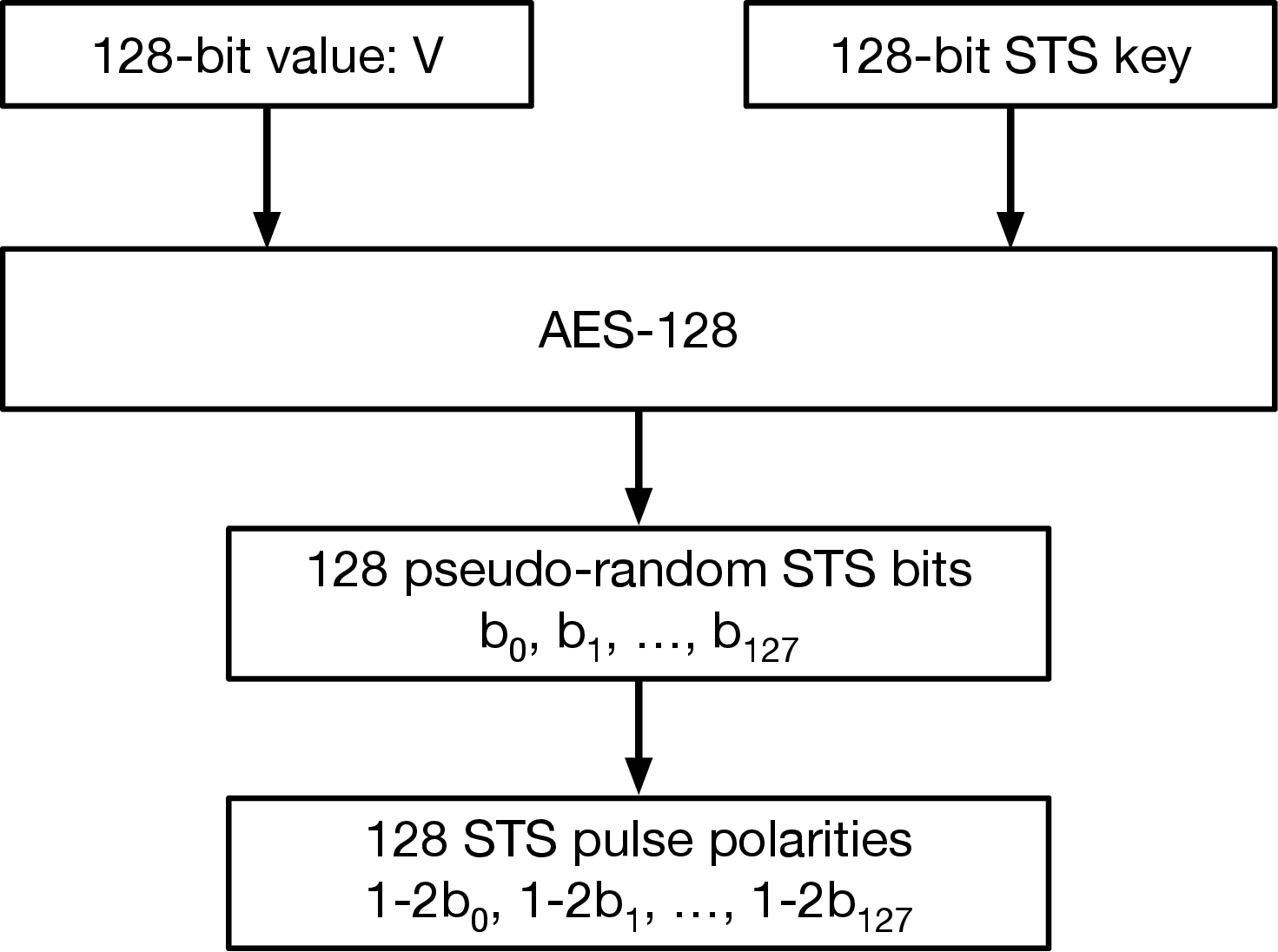, width=0.8\linewidth}
\caption{STS generation with AES-128 in counter mode.}
\label{fig:DRBG}
\end{figure}

\vspace{0mm}
The deterministic random bit generator (DRBG) \cite{DRBGstd} for generation
of STS bits is specified by IEEE in \cite{802.15.4z} and is depicted in
Fig. \ref{fig:DRBG}.
This DRBG produces $128$ pseudo-random STS bits after each run. The higher
layer is responsible for setting the $128$-bit STS key and the $128$-bit
initial value for $V$. The value of $V$ is incremented after each run of the
DRBG and outputting $128$ bits. In this way, fresh STS bits are generated
for each STS packet.

\vspace{0mm}
IEEE also recommends that re-seeding should be carried out when specific
levels of backtracking resistance are required or a large number of
iterations are performed.

\section{Proof of Proposition \ref{PropCIREst}}\label{AppendixCIREst}
\begin{proof}
We can re-write the received signal in (\ref{rxSig3}) as
\begin{eqnarray}
r[n] &=& \sum_{k=0}^{(Q-1)M} s_e[k]g[n-k] + w[n],\\
     &=& \sum_{k=0}^{JM-1} g[k]s_e[n-k] + w[n], \label{expandedSig}
\end{eqnarray}
where $s_e[k]$ is an expanded version of $s[k]$, i.e.,
$s_e[k] = s[k/M]$, $\forall k=0,\pm M, \pm 2M, ...$ and
$s_e[k]=0$ otherwise, and the second equality is due to the
commutative property of the linear convolution operation.

By utilizing the vector notation as defined in
Proposition \ref{PropCIREst}, the signal model
in (\ref{expandedSig}) can be further expressed as
\begin{equation}
{\bm r} = {\bm \Phi} {\bm g} + {\bm w}, \label{vecSig}
\end{equation}
where ${\bm w}:=[w[0], ...,w[(Q-1+J)M-1]]^T$.
The least-squares solution to the linear system in (\ref{vecSig})
can be obtained as in (\ref{slcCIR}), i.e., 
\begin{equation}
\hat{\bm g} 
= \left({\bm \Phi}^T{\bm \Phi}\right)^{-1}{\bm \Phi}^T{\bm r}.
\end{equation}
Note that the least-squares solution is indeed the maximum-likelihood
one when the noise vector $\bm w$ is white Gaussian.
\end{proof}

\section{CIR Thresholding}\label{AppendixCIRThr}
One CIR tap could be validated as a true one when the amplitude
of the tap exceeds a certain acceptance threshold with respect to
the learned standard deviation according to the desired false
acceptance rate. In Fig. \ref{fig:CIRThreshold}, the
threshold is set to $3\sigma$ to realize a false acceptance rate
around $0.1\%$ under the Gaussian distribution.
On the other hand, fake peaks can show up earlier than the
true first path in the presence of attacks.

\begin{figure}[h]
	\centering
	\epsfig{file=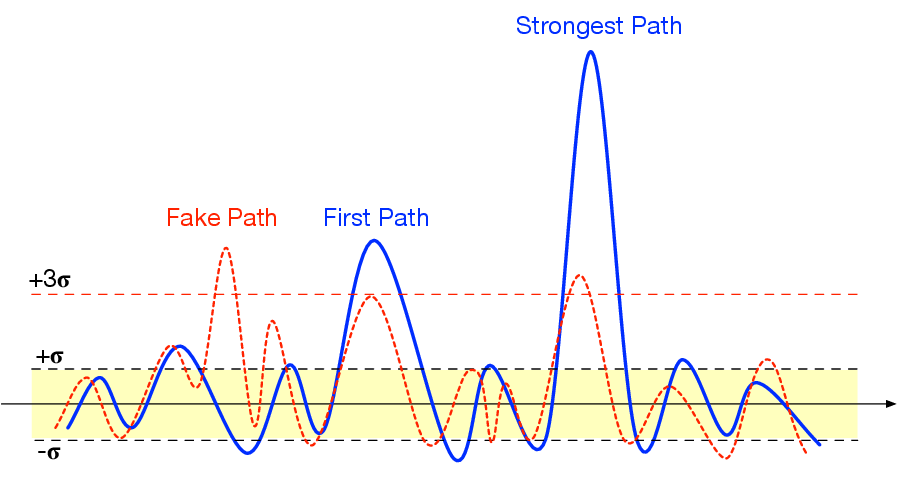, width=0.85\linewidth}
	\caption{Solid \textcolor{blue}{\bf Blue}: estimated
		CIR in the absence of adversarial attacks.
		Dotted \textcolor{red}{\bf Red}: CIR with fake peaks in the presence
		of attacks as in \cite{2022GhostPeak,Capkun2021_HRP}. The standard
	    deviation of the estimation errors is $\sigma$.}
	\label{fig:CIRThreshold}
\end{figure}

\section{Proof of Proposition \ref{PropAdapAttk}}\label{AppendixAdapAttk}
\begin{proof}
Among all the received signal samples in (\ref{rxSigAttk2}),
we can focus on the following set of signal samples that
convey information about the tap
$l_* = \bar{\tau}_0 + \epsilon - \Lambda M$:
\begin{equation}\label{yn_attk}
\begin{split}
y[n] := r_a[l_*+nM] 
= \theta \sum_{k} p_R[k-\bar{\tau}_0]a_e[l_*+nM-k],
\end{split}
\end{equation}
where $a_e[k]$ the expanded version of $a[k]$, i.e., $a_e[k]=a[k/M]$
when $k$ is a multiple of $M$ and $a_e[k]=0$ otherwise. Considering
$\epsilon\in[0, M-1]$, we can further simplify $y[n]$ in
(\ref{yn_attk}) as 
\begin{eqnarray}
y[n] &=& \theta \hspace{-4mm}\sum_{{k},{k=l_*+ jM}}
\hspace{-4mm} p_R[k-\bar{\tau}_0]a_e[l_*+nM-k]\nonumber\\ 
&=&\theta\sum_{j} p_R[l_*+jM-\bar{\tau}_0]a_e[(n-j)M]   \nonumber\\
&=& \theta \sum_{j} p_R[(j-\Lambda)M+\epsilon] a[n-j] \nonumber\\
&=& \theta \cdot p_R[\epsilon]\cdot a[n-\Lambda], \label{selSignal}
\end{eqnarray}
where we have utilized the assumption that $p_R(t)$ is confined within the
time interval of $[0, MT_0)$ to get the last equality.

With the observations $\{y[n]\}$ in (\ref{selSignal}), the RXD first
obtains an initial estimate of the CIR tap at $l_*+\zeta M$,
$\zeta\ge 0$, according to (\ref{estCIRattk1}) as
\begin{equation}\label{initEst}
	\widetilde{\sf CIR}_{l_*+\zeta M} = 
	\frac{1}{Q}\sum_{n=0}^{Q-1}s[n-\zeta]\cdot y[n].
\end{equation}
Note that $s[n]=0$ when $n<0$.
In particular, after substituting $y[n]$ with (\ref{selSignal}) and
following the definition of attack in (\ref{attkSymb}), the
estimate of the CIR tap at $l_*+(\Lambda+\lambda) M$ can be found in
(\ref{initEst2}).

\begin{figure*}[htp!]
\begin{equation}
\label{initEst2}
\begin{aligned}
&\widetilde{\sf CIR}_{l_*+(\Lambda+\lambda) M} \\
&=\frac{\theta \cdot p_R[\epsilon]}{Q}
\sum_{n=0}^{Q-1}s[n-\Lambda-\lambda]\cdot a[n-\Lambda]
= \frac{\theta \cdot p_R[\epsilon]}{Q}\sum_{n=0}^{Q-1}s[n-\Lambda-\lambda]\cdot
\left(
\sum_{\lambda'=0}^{H}s[n-\Lambda-\lambda']
\sum_{l=\Lambda+H}^{n-\Lambda}s[l]s[l-\Lambda-\lambda']
\right)\\
&= \frac{\theta \cdot p_R[\epsilon]}{Q} 
\sum_{\lambda'=0}^{H}\sum_{n=0}^{Q-1}\sum_{l=\Lambda+H}^{n-\Lambda}
s[n-\Lambda-\lambda]s[n-\Lambda-\lambda']\cdot s[l]s[l-\Lambda-\lambda']
\end{aligned}
\end{equation}
\hrulefill
\end{figure*}

To refine the CIR estimate in (\ref{initEst2}), interference cancellation
can be performed as in (\ref{estCIRattk2}). As the first-order approximation
to the matrix inversion in (\ref{estCIRattk2}), we have
\begin{equation}
\left({\bm \Phi}^T{\bm \Phi}/Q\right)^{-1} = 
\left({\bm I}_{JM} + {\bm \Delta} \right)^{-1}
\approx {\bm I}_{JM} - {\bm \Delta},
\end{equation}
where ${\bm I}_{JM}$ denotes the $JM\times JM$ identity matrix and
$\bm \Delta$ captures the off-diagonal autocorrelation of the STS
sequence. Specifically, we have ${\bm \Delta}(i,i)=0$ and
${\bm \Delta}(i,i+j) = \sum_k s_e[k]s_e[k-j]/Q$ when $j\ne 0$,
where $s_e[k]$ is defined as in (\ref{expandedSig}).
According to the above approximation, we can cancel the interference
to CIR tap at $l_*$ as
\begin{equation}\label{refinedEst}
\begin{split}
&\widehat{\sf CIR}_{l_*} = \frac{1}{Q}\cdot \\
&~~~\sum_{n=0}^{Q-1}s[n]\cdot
\left(
y[n]- \sum_{\lambda\ne -\Lambda}
s[n-\Lambda-\lambda]\cdot \widetilde{\sf CIR}_{l_*+(\Lambda+\lambda) M} 
\right)\\
&= \widetilde{\sf CIR}_{l_*} - \sum_{\lambda\ne -\Lambda}\sum_{n=0}^{Q-1}
\frac{\widetilde{\sf CIR}_{l_*+(\Lambda+\lambda)M}}{Q}
 s[n]s[n-\Lambda-\lambda].
\end{split}
\end{equation}

By utilizing the IID property of the STS sequence, we can easily
verify ${\sf E}[\widetilde{\sf CIR}_{l_*}] = 0$ and further
obtain the result in (\ref{meanCIR1}), where we have used the fact
that the expected value is zero whenever $\lambda'\neq \lambda$
to get the second equality. 

\begin{figure*}[htp!]
\begin{equation}
\label{meanCIR1}
\begin{aligned}
{\sf E}\left[\widehat{\sf CIR}_{l_*}\right] 
&=
-\frac{\theta \cdot p_R[\epsilon]}{Q^2}\cdot\sum_{\lambda\ne -\Lambda}
{\sf E}\left[
\sum_{n=0}^{Q-1} s[n]s[n-\Lambda-\lambda]
\sum_{\lambda'=0}^{H}
\sum_{n'=0}^{Q-1}
\sum_{l=\Lambda+H}^{n'-\Lambda}
s[n'-\Lambda-\lambda]s[n'-\Lambda-\lambda']
s[l]s[l-\Lambda-\lambda']
\right] \\
&= 
-\frac{\theta \cdot p_R[\epsilon]}{Q^2}\cdot\sum_{\lambda=0}^{H}
{\sf E}\left[
\sum_{n=0}^{Q-1} s[n]s[n-\Lambda-\lambda]
\sum_{n'=0}^{Q-1}
\sum_{l=\Lambda+H}^{n'-\Lambda}
s[l]s[l-\Lambda-\lambda]
\right] \\
&= 
-\frac{\theta \cdot p_R[\epsilon]}{Q^2}\cdot\sum_{\lambda=0}^{H}
{\sf E}\left[
\sum_{n'=0}^{Q-1}
\sum_{l=\Lambda+H}^{n'-\Lambda}
s[l]s[l-\Lambda-\lambda]
\sum_{n=0}^{Q-1} s[n]s[n-\Lambda-\lambda]
\right]
\end{aligned}
\end{equation}
\hrulefill
\end{figure*}

By resorting to the IID property of the STS sequence again,
we have
\begin{equation}\label{meanCIR2}
{\sf E}\left[\widehat{\sf CIR}_{l_*}\right] 
= -\frac{\theta \cdot p_R[\epsilon]}{Q^2}\cdot\sum_{\lambda=0}^{H}
\sum_{n'=2\Lambda+H}^{Q-1}(n'-2\Lambda-H + 1).
\end{equation}
Accordingly, we can obtain the following asymptotic result:
\begin{equation}
\lim_{Q\rightarrow\infty} {\sf E}\left[\widehat{\sf CIR}_{l_*}\right]
=-\frac{H+1}{2}\theta \cdot p_R[\epsilon].
\end{equation}
\end{proof}

\section{Azuma-Hoeffding Inequality}\label{AppendixHoeffding}
The Azuma-Hoeffding inequality and the proof can be found in
\cite{1963Hoeffding,RandomProcbook}. We summarize this important
inequality here for convenient reference.

{\bf Azuma-Hoeffding Inequality}: 
{\it Let $\{X_0, X_1, ...\}$ be a martingale, and suppose that there
exists two sequence of real numbers $\{A_0, A_1, ...\}$ and
$\{B_0, B_1, ...\}$ such that
${\sf Pr}(A_t\le X_t-X_{t-1}\le B_t) = 1$, $\forall t$. 
Then we have, $\forall \epsilon>0$,
\begin{equation}
{\sf Pr}(X_n-X_0 \ge \epsilon) \le
\exp
\left(
-\frac{2\epsilon^2}{\sum_{t=1}^n (B_t - A_t)^2}
\right).	
\end{equation}
}

\section{Sampling Clock at RXD}\label{AppendixRXDSampling}
At shown in Fig. \ref{fig:sysModel}, the RXD needs to carry out
timing and frequency synchronization to the received signal before
processing it further. When there is only the legitimate signal
from the TXD in the air, the RXD will acquire the timing and
frequency of the TXD. On the other hand, the acquired timing and
frequency at the RXD will sync with the attack signal when the
transmission from an adversary dominates.

\vspace{4mm}
In the case of attack, the RXD will sample and process the received
signal according to the the clock rate of the attacker upon acquiring
timing and frequency. Specifically, the RXD will utilize $\tilde{T}_c$,
the chip duration from the perspective of the adversary, to
determine the sampling period in (\ref{rxSig3}), i.e.,
$T_0 = \tilde{T}_c/\Omega$. Note that we have $\tilde{T}_c > T_c$
when the attacker runs a slow clock as described in Section
\ref{sSec:ClockOffsetAttk}.

\vspace{4mm}
We have looked into the scenario where the clock at the attacker
ticks at the same rate as the TXD clock in Section \ref{sSec_Security}.
In this case, the RXD samples the received attack signal
with the clock rate matching that of the TXD. We have further
examined how an attacker could take advantage of the clock offsets
relative to the TXD in Section \ref{sSec:ClockOffsetAttk}.

\newpage
\section{Proof of Proposition \ref{PropDetperf}}\label{AppendixDetPerf}
\begin{proof}
	Let $G_n := \sum_{k=0}^{n-1} ({\sf E}[x[k]s[k]]-x[k]s[k])$. We have
	\begin{equation}
		\begin{split}
			&{\sf E}[G_{n+1}|G_n,G_{n-1}, ..., G_0] \\
			&= G_n + {\sf E}\big[{\sf E}[x[n]s[n]]-x[n]s[n] \big| G_n,G_{n-1}, ..., G_0\big]\\
			&= G_n + {\sf E}[x[n]s[n]] - {\sf E}\big[x[n]s[n] \big| G_n,G_{n-1}, ..., G_0\big]\\ 
			&= G_n + {\sf E}[{\mathbb M}(s[n])s[n]] + {\sf E}[e[n]s[n]]\\
			&\hspace{+9.5mm}-{\sf E}\big[{\mathbb M}(s[n])s[n] \big| G_n,G_{n-1}, ..., G_0\big]\\ 
			&\hspace{+9.5mm}-{\sf E}\big[e[n]s[n] \big| G_n,G_{n-1}, ..., G_0\big]\\ 
			&\stackrel{(a)}{=} G_n - {\sf E}\big[e[n]s[n] \big| G_n,G_{n-1}, ..., G_0\big] \stackrel{(b)}{=} G_n
		\end{split},
	\end{equation}
	where equality (a) utilizes the result in (\ref{uncorrelatedError}) and
	also the fact that $s[n]$ is independent of $\{x[k], s[k]\}_{k<n}$ when
	the tap $l_*$ corresponds to the true first path;
	equality (b) is due to the assumption that the estimation error
	$e[n]$ is uncorrelated with $s[n]$ conditioned on $\{x[k]\}_{k<n}$ and
	$\{s[k]\}_{k<n}$. Further, we have
	\begin{equation}
		A_n \le G_n-G_{n-1} \le B_n,
	\end{equation}
	where 
	\begin{eqnarray}
		A_n = {\sf E}[{\mathbb M}(s[n-1])s[n-1]]-1;   \\
		B_n = {\sf E}[{\mathbb M}(s[n-1])s[n-1]]+1;
	\end{eqnarray}
	By resorting to the Azuma-Hoeffding inequality in Appendix
	\ref{AppendixHoeffding} again, $\forall \epsilon>0$,
	we have
	\begin{equation}
		{\sf Pr}(G_Q \ge \epsilon) \le \exp\left(-\frac{\epsilon^2}{2Q}\right).
	\end{equation}
	The miss detection performance of the detector in (\ref{validation})
	is thus given by
	\begin{equation}
		\begin{split}
			&{\sf Pr}(T({\bm x}) < \gamma) \\ 
			&= {\sf Pr}\left(G_Q > \sum_{n=0}^{Q-1}
			{\sf E}[{\mathbb M}(s[n])s[n]]-\sqrt{Q}\cdot \gamma\right)\\
			&= {\sf Pr}\left(G_Q > Q\bar{C}-\sqrt{Q}\cdot \gamma\right)\\
			&\le \exp\left(-\frac{Q(\bar{C}-\gamma/\sqrt{Q})^2}{2}\right),
		\end{split}
	\end{equation}
	where $\bar{C}$ is as defined in Proposition \ref{PropDetperf}.
\end{proof}

\begin{figure}[t]
	\centering
	\epsfig{file=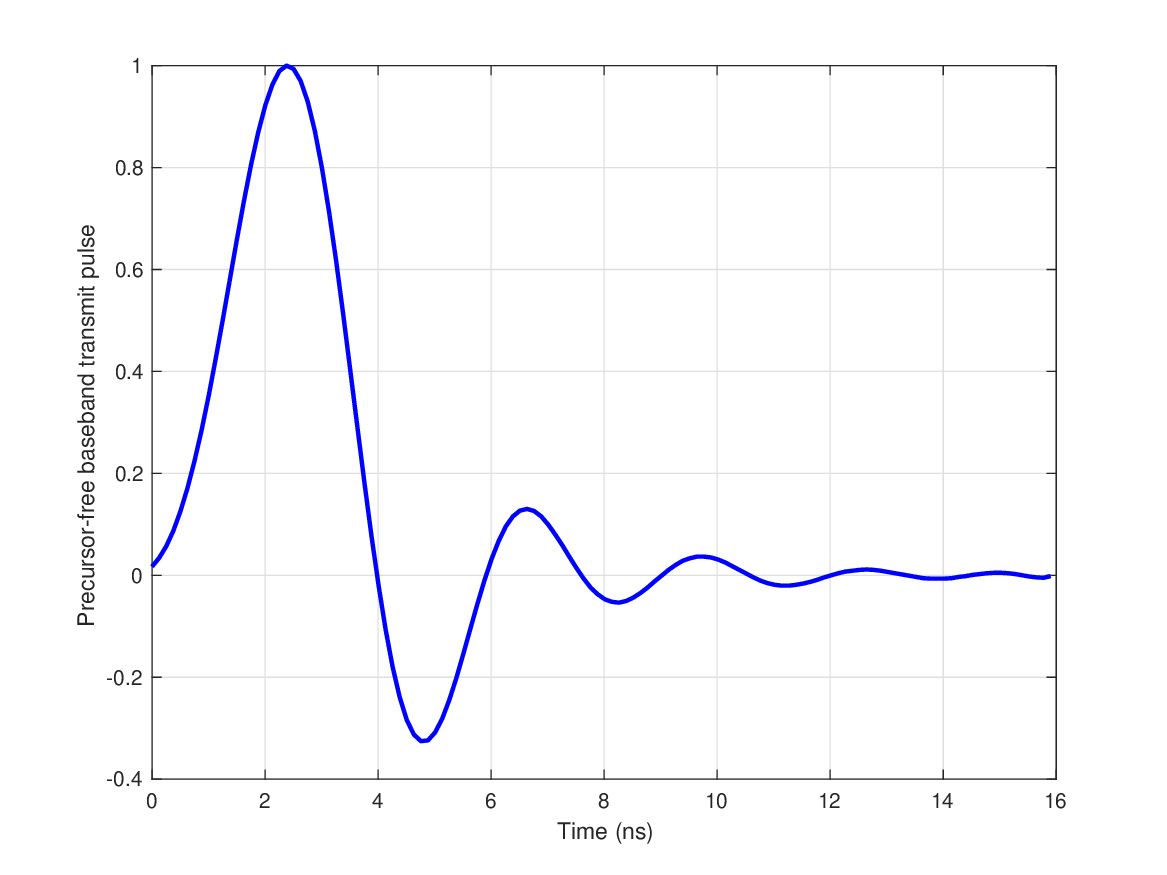, width=0.85\linewidth}
	\caption{Baseband pulse with minimum precursor energy.}
	\label{fig:minPhaseSRRC}
\end{figure}

\section{Precursor-Free Pulse}\label{AppendixPCF}
In IEEE 802.15.4 \cite{802.15.4}, a square-root raised cosine (SRRC) pulse is
defined as follows. 
\begin{equation}
	p_{T, {\sf srrc}}(t) = \frac{2}{\pi\sqrt{T_c}}
	\frac{\cos\left(\frac{3\pi t}{2T_c}\right)+\frac{\sin(0.5\pi t/T_c)}{2t/T_c}}
	{1-(2t/T_c)^2}, \label{srrc}
\end{equation}
where $T_c=1/499.2$ $\mu$s denotes the length of a chip. The minimum phase
decomposition \cite{DSPbook} of the SRRC pulse in (\ref{srrc}) is plotted in
Fig. \ref{fig:minPhaseSRRC}. Note that this minimum phase pulse exhibits maximum
energy concentration at the pulse beginning.


\end{document}